\documentclass[journal,12pt,onecolumn,draftclsnofoot]{IEEEtran}
\normalsize

\usepackage{color}
\usepackage{graphicx}
\usepackage{amsfonts}
\usepackage[cmex10]{amsmath}
\usepackage{cite}
\usepackage{algorithmic}
\usepackage[center]{caption}
\usepackage{epsfig}
\usepackage{latexsym}
\usepackage{epstopdf}
\usepackage{verbatim}
\usepackage{amsthm}
\usepackage[linesnumbered,ruled,vlined]{algorithm2e}
\usepackage{bbm}
\usepackage{chngcntr}

\usepackage{mathtools}
 
\graphicspath{{./figures/}}
\theoremstyle{definition}
\newtheorem{definition}{Definition}
\newtheorem{example}{Example}

\newtheorem{construction}{Construction}
\newtheorem{theorem}{Theorem}

\newtheorem{remark}{Remark}
\newtheorem{lemma}[theorem]{Lemma}
\newcounter{numrel}
\counterwithin*{numrel}{section}

\newcommand{\remove}[1]{}

 \usepackage{amssymb}
\begin{document}
	%
\title{Harnessing Correlations in Distributed Erasure-Coded Key-Value Stores} 

\author{Ramy E. Ali and Viveck R. Cadambe
	\thanks{Ramy E. Ali (email: ramy.ali@psu.edu) and  Viveck R. Cadambe (email: viveck@engr.psu.edu) are with the School of Electrical Engineering and Computer Science, The Pennsylvania State University, University
				Park, PA 16802. This work is supported by NSF grant No. CCF 1553248 and is published in part in the Proceedings of the 2016 IEEE Information Theory Workshop \cite{ali2016consistent}.}
	}

%

\maketitle
\begin{abstract}
Motivated by applications of distributed storage systems to key-value stores, the multi-version coding problem was formulated to efficiently store frequently updated data in asynchronous decentralized storage systems. Inspired by consistency requirements in distributed systems, the main goal in the multi-version coding problem is to ensure that the latest possible version of the data is decodable, even if the data updates have not reached some servers in the system. In this paper, we study the storage cost of ensuring consistency for the case where the data versions are correlated, in contrast to previous work where data versions were treated as being independent. We provide multi-version code constructions that show that the storage cost can be significantly smaller than the previous constructions depending on the degree of correlation, despite the asynchrony and the decentralized nature. Our achievability results are based on Reed-Solomon codes and random binning. Through an information-theoretic converse, we show that our multi-version codes are nearly-optimal, within a factor of $2$, in certain interesting regimes.

\end{abstract}
\IEEEpeerreviewmaketitle

\section{Introduction}
Distributed key-value stores are an important part of modern cloud computing infrastructure\footnote{A \emph{read-write} key-value store is a shared database that supports \emph{get} or read operations, and \emph{put} or write operations.}. Key-value stores are commonly used by several applications including reservation systems, financial transactions, collaborative text editing and multi-player gaming. Owing to their utility, there are numerous commercial and open-source cloud-based key-value store implementations such as Amazon Dynamo \cite{Decandia}, Google Spanner \cite{corbett2013spanner} and Apache Cassandra \cite{CassandradB}. Distributed data storage services including key-value stores commonly build fault tolerance and availability by replicating the data across multiple nodes. Unlike archival data storage services, in key-value stores, the data stored is updated frequently, and the time scales of data updates and access are often comparable to the time scale of dispersing the data to the nodes \cite{Decandia}. In fact, distributed key-value stores are commonly \emph{asynchronous} as the time scale of data propagation is unpredictable and different nodes can receive the updates at different points in time. In such settings, ensuring that a client gets the latest version of the data requires careful and delicate protocol design.   

 Modern key-value stores depend on  high-speed memory that is expensive as compared with hard drives to provide fast operations. Hence, the goal of efficiently using memory has motivated significant interest in erasure-coded key-value stores. In absence of consistency requirements, erasure-coded ``in-memory'' storage systems can significantly improve latency as compared to replication-based counter parts \cite{EC-Cache}. Systems research related to erasure-coded consistent data stores also has received significant interest (See, e.g., \cite{taranov2018fast} and also \cite{Giza}, which uses erasure coding for Microsoft's data centers even for non-archival consistent data storage services). Delta coding  \cite{deltacompression} is another classical technique that improves memory efficiency in data storage systems, where it is desired to store older versions of the data. Delta coding relies on the idea of compressing differences between subsequent versions to improve the storage cost.  

The main contribution of our paper is developing a coding-theoretic approach that combines erasure coding and delta coding to exploit correlations between subsequent updates and enable significant memory savings as compared to replication-based schemes and erasure coding approaches that do not exploit correlations. We begin by an overview of the consistent data storage algorithms, and then discuss the \emph{multi-version coding} framework, which is an information-theoretic framework for distributed storage codes tailor-made for consistent key-value stores.
\subsection{Overview of Key-Value Stores} 
The design principles of key-value stores are rooted in the distributed computing-theoretic abstraction known as \emph{shared memory emulation} \cite{Lynch1996}. The goal of the \emph{read-write} shared memory emulation problem is to implement a read-write variable over a distributed system. While there has been interest in archival data storage systems in coding theory, e.g. \cite{Dimakis, Tamo_Barg}, the shared memory emulation set up differs in the following aspects.
\begin{enumerate}
	\item \emph{Asynchrony:} a new version of the data may not arrive at all servers simultaneously.
	\item \emph{Decentralized nature:} there is no single encoder that is aware of all versions of the data simultaneously, and a server is not aware of which versions are received by other servers. 
\end{enumerate}
	 Shared memory emulation algorithms use \emph{quorum-based} techniques to deal with the asynchrony. Specifically, in a system of $n$ servers that tolerates $f$ crash-stop server failures \cite{Lynch1996}, a write operation sends a request to all servers and waits for the acknowledgments from any $c_W\leq n-f$ servers for the operation to be considered \emph{complete}. Similarly, a read operation sends a request to all servers and waits for the response of any $c_R \leq n-f$ servers. This strategy ensures that, for every complete version, there are at least $c \coloneqq c_W+c_R-n$ servers that received that version and responded to the read request. Shared memory emulation algorithms require that the \emph{latest} complete version, or a later version, must be recoverable by the read operation. This requirement is referred to as consistency\footnote{More specifically, our decoding requirement is inspired by the consistency criterion known as \emph{atomicity}, or \emph{linearizability}.} in distributed systems \cite{Lynch1996}. 
	 
	 In a replication-based protocol, where servers store an uncoded copy of the latest version that they receive, selecting $c_W$ and $c_R$ such that $c_W+c_R>n$ ensures consistency \cite{Decandia, Lynch1996}. Since for every complete write operation, there are $c$ servers that store the value of that write operation and respond to a given read operation, it seems natural to use a maximum distance separable (MDS) code of dimension $c$ to obtain storage cost savings over replication-based algorithms. However, the use of erasure coding in asynchronous distributed systems where consistency is important leads to interesting algorithmic and coding challenges. This is because, when erasure coding is used, no single server stores the data in its entirety; for instance, if an MDS code of dimension $c$ is used, each server only stores a fraction of $1/c$ of the entire value. Therefore, for a read operation to get a version of the data, at least $c$ servers must send the codeword symbols corresponding to this version. As a consequence, when a write operation updates the data, a server cannot delete the symbol corresponding to the old version before symbols corresponding to a new version propagate to a sufficient number of servers. That is, servers cannot simply store the latest version they receive; they have to store older versions as shown in Fig. \ref{MDS}. 
\begin{figure}[h]
	\centering
	\includegraphics[width=.65\textwidth,height=.25\textheight]{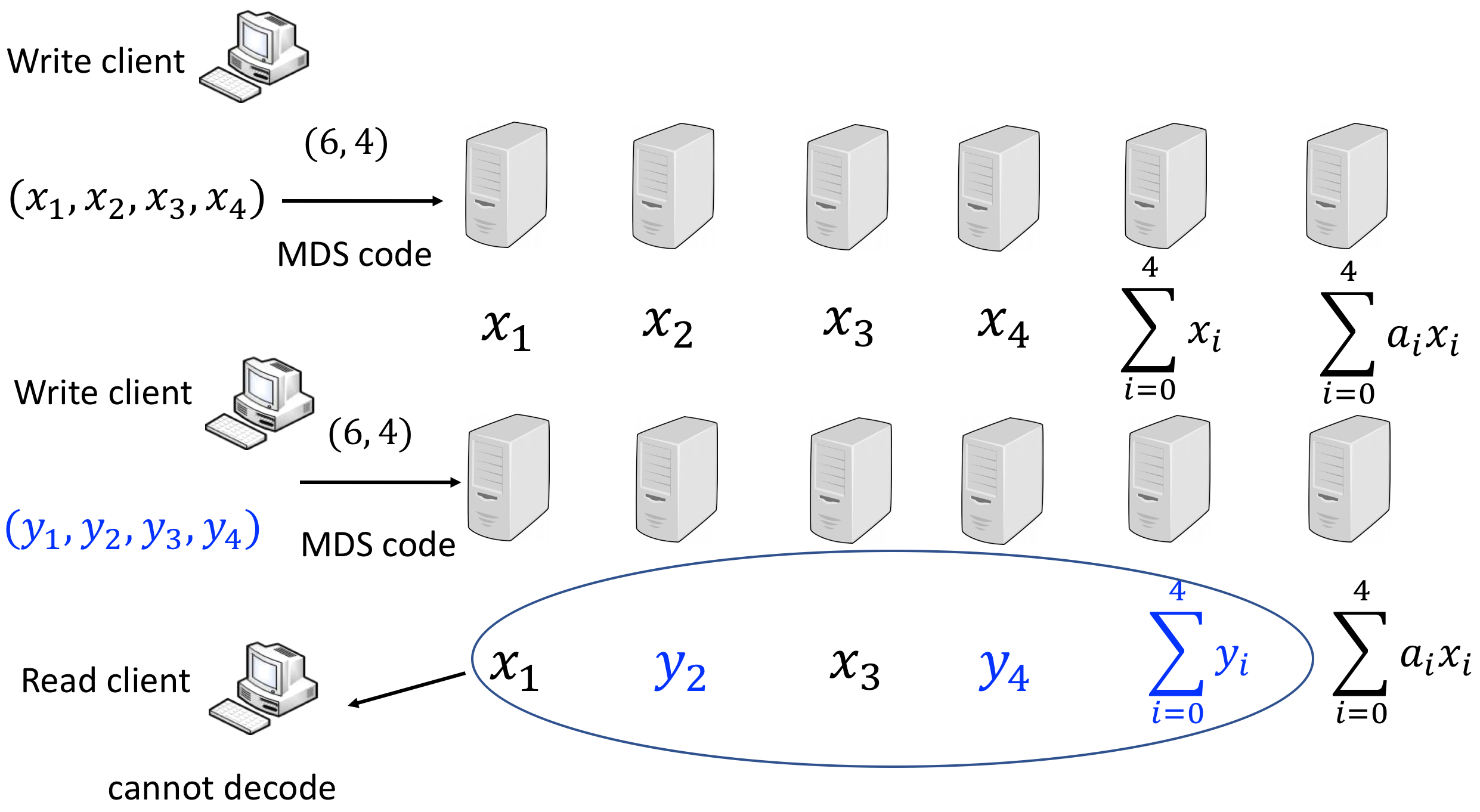}
	\caption{An erasure-coded algorithm where the servers only store the codeword symbol of the latest version is considered, where $n=6, c_W=5$ and $c_R=5$. The old value of the variable is $(x_1, x_2, x_3, x_4)$ and is updated to $(y_1, y_2, y_3, y_4)$. Since all servers may not receive the new codeword symbols simultaneously, a read operation may not be able to decode. \label{MDS}}
\end{figure}
Given that storing multiple versions is inevitable in consistent erasure-coded systems \cite{CadambeCoded_NCA, Dutta, multi_arxiv}, an important opportunity to improve memory efficiency is to exploit correlations between the various versions; this opportunity is the main motivation of our paper. We conduct our study through the \emph{multi-version coding} framework \cite{multi_arxiv}.
\color{black}
\subsection{Multi-Version Coding}
The \emph{multi-version coding} problem abstracts out algorithmic details of shared memory emulation while retaining the essence of consistent storage systems. Specifically, the multi-version coding problem \cite{multi_arxiv} considers a decentralized storage system with $n$ servers that tolerates $f$ crash-stop failures, where the objective is storing a message (read-write variable) of length $K$ bits with $\nu$ versions\footnote{We study the case where each version is $K$ bits long; although certain applications might benefit from dynamic allocation, several important applications use a fixed size for various versions of the value. Furthermore, popular key-value stores expose a fixed-sized value to the client and do not expose ``malloc''-type dynamic allocation.}. The versions are totally ordered; versions with higher ordering are referred to as later versions, and lower ordering as earlier versions. Each server receives an arbitrary subset of the $\nu$ versions, and encodes them. Because of the decentralized nature, a server is unaware of which versions are available at other servers. Inspired by the quorum-based protocols, we refer to any version that has reached at least $c_W$ servers as a \emph{complete} version. A decoder connects to any $c_R$ servers, and must recover the latest complete version, or a later version.

In \cite{multi_arxiv}, it was shown that there is a storage cost inevitable price for maintaining consistency in asynchronous decentralized storage systems. In multi-version coding, for any complete version, for any decoder, there are at least $c$ servers that have received that version and responded to the decoder. In the classical erasure coding model, where $\nu=1$, the Singleton bound dictates that the storage cost per server is at least $K/c$. However for $\nu > 1$, a server cannot simply store the codeword symbol corresponding to one version. In the case where the versions are independent, it was shown in \cite{multi_arxiv} that the storage cost per server is at least
$\frac{\nu}{c+\nu-1}K-\Theta(1)$. Since, for $\nu < c,$ we have $\frac{\nu}{c+\nu-1} \geq \frac{\nu}{2c}$, and  since the per-server cost of storing a version is $K/c$, we may interpret the result as follows: when the versions are independent, to compensate for the asynchrony and maintain consistency, a server has to store an amount of data that is, from a cost perspective, tantamount to at least $\nu/2$ versions, each stored using an MDS code of dimension $c$. 

Although the study of \cite{multi_arxiv} focuses on coding-theoretic aspects, the insights obtained from this study have been incorporated into distributed algorithms \cite{ zorgui2018storage} and a  lower-bound was developed in \cite{Cadambe_Wang_Lynch2016} on the storage cost of any read-write memory emulation algorithm by creating a worst-case execution mimicking the converse of \cite{multi_arxiv}. We believe that merging the coding-theoretic ideas in our paper and the algorithmic insights of \cite{zorgui2018storage} is an interesting area of future work.

\subsection{Contributions}
In this paper, we extend the scope of the multi-version coding problem to the case where the different versions are correlated. Specifically, we consider a decentralized storage system with $n$ servers storing $\nu$ possibly correlated versions of a message. We assume that each message version is $K$ bits long, and model the correlation between successive versions in terms of the bit-strings that represent them. Given a version, we assume that the subsequent version is uniformly distributed in the Hamming ball of radius $\delta_KK$ centered around that given version. Hence, this version can be represented using $\log Vol(\delta_K K, K)$ bits, where $Vol(\delta_K K, K)$ is the volume of the Hamming Ball of radius $\delta_K K$. We derive three main results for this system.
\begin{enumerate}
	\item We first study the case where $\delta_K$ is not known and propose an effective scheme based on Reed-Solomon code with a per-server storage cost of $\frac{K}{c}+   (\nu-1) \delta_K K (\log K + o(\log K))$ bits. This scheme obtains the $1/c$ \emph{erasure coding gain} for the first version and stores every subsequent version via delta coding with a cost of $\delta_K K (\log K+o(\log K))$ bits. Thus, this scheme is unable to simultaneously obtain the gains of both erasure and delta coding.

 \item We then study the case where $\delta_K$ is known and derive a random binning based scheme with a per-server storage cost of $\frac{K}{c}+\frac{\nu -1}{c} {\log Vol(\delta_K K, K)} + o(\log K)$ bits. From a cost viewpoint, this scheme is tantamount to storing one version using erasure coding with a cost of $K/c$ and performing delta and erasure coding for the subsequent versions leading to a cost of $\frac{\log Vol(\delta K, K)}{c}$ bits per version. This scheme outperforms our first scheme as it simultaneously obtains the gains of both erasure coding and delta coding for subsequent versions. We also show the existence of linear codes that obtain this storage cost. 

A cost of $\frac{K}{c}+\frac{\nu -1}{c} {\log Vol(\delta_K K, K)} + o(\log K)$ bits is readily achievable in a setting, where every server receives all the versions, and each server is aware that the other servers have indeed received all the versions. In such a setting, each server can store a fraction of $1/c$ of the first version it receives using an MDS code of dimension $c$. For a new version, each server can store $1/c$ of the compressed difference between this version and the old version using an MDS code of dimension $c$. However, this scheme would fail in our setting because of the decentralized asynchronous nature. For instance, a server which receives versions $1$ and $3$ needs to compress version $3$ with respect to version $1$ and then encodes it, but a different server that receives only versions $2$ and $3$ needs to compute the increment of version $3$ with respect to version $2$ and then encodes it; from a decoder's viewpoint, the erasure-coded symbols stored at the two servers are not compatible. Furthermore, the decentralized nature implies that the server that receives versions $1$ and $3$ must store some data that would enable successful decoding no matter what versions are received by the other servers. 
Handling the decentralized and asynchronous nature while achieving both the erasure and the delta coding gain is our main technical contribution. 

\item We extend the lower bound of \cite{multi_arxiv} to the case of correlated versions and show our random binning scheme is within a factor $2$ of the information-theoretic optimum in certain regimes.

\end{enumerate}
\subsection{Related Work}
The idea of exploiting the correlation to efficiently update, store or exchange data has a rich history of study starting from the classical Slepian-Wolf problem \cite{slepian1973noiseless,wyner1974recent, pradhan2003distributed, schonberg2004distributed, schonberg2007practical} for compressing correlated distributed sources. Encoding incremental updates efficiently is the motivation of the delta compression techniques used commonly in data storage. The notion of delta compression was refined in \cite{wang2015file,ma2012compression} by modeling the data updates using the edit distance; in particular, these references develop schemes that synchronize a small number of edits between a client and a server efficiently. While we note that the edit distance is relevant to applications such as collaborative text editing, we focus on the classical Hamming metric used more widely in coding theory for the sake of fundamental understanding and other applications that view the data as a table as in Apache Cassandra \cite{CassandradB}, and the writes update only a few entries of the table.

Exploiting correlations to improve efficiency in \emph{distributed} storage and caching settings has been of significant interest \cite{harshan2015compressed, milenkovic_updates, anthapadmanabhan2010update, mazumdar2014update, nakkiran2014fundamental, Prakash_Medard, Tulino_Correlated}. In \cite{harshan2015compressed} and \cite{milenkovic_updates}, coding schemes were developed that use as input, the old and the new version of the data, and output a code that can be used to store both versions efficiently. Capacity-achieving \emph{update-efficient} codes for binary symmetric and erasure channels were studied in \cite{anthapadmanabhan2010update,mazumdar2014update}, where a small change in the message leads to a codeword which is close to the original codeword in Hamming distance. In \cite{nakkiran2014fundamental}, the problem of minimizing the communication cost of updating a ``stale'' server that did not get an updated message, by downloading data from already updated servers, was studied and constructions and tight bounds were developed. A side information problem is presented in \cite{Prakash_Medard}, where the goal is to send an updated version to a remote entity that has as side information an arbitrary linear transform of an old version. The reference shows that the optimal encoding function is related to a maximally recoverable subcode of the linear transform associated with the side information. 

Although our problem has common ingredients with previous works, our setting differs as it captures the asynchrony, decentralized nature, and the consistency requirements. An important outcome of our study is that correlation can reduce storage costs despite these requirements.

\subsection*{Organization of the paper}
This paper is organized as follows. Section \ref{Model} presents the multi-version coding problem and the main results of this paper. In Section \ref{Code Constructions}, we provide our code constructions. Section \ref{Lower Bound on the storage cost} provides a lower bound on the storage cost. Finally, conclusions are discussed in Section \ref{Conclusion}.

\section{System Model and Background of Multi-version Codes}
\label{Model}

We start with some notation. We use boldface for vectors. In the $n$-dimensional space over a finite field $\mathbb{F}_p$, the standard basis column vectors are denoted by $\{ \mathbf e_1, \mathbf e_2, \cdots, \mathbf e_n \}$\color{black}. We denote the Hamming weight of a vector $\mathbf{x}$ by $w_H(\mathbf{x})$ and the Hamming distance between any two vectors $\mathbf{x_1}$ and $\mathbf{x_2}$ by $d_H(\mathbf{x_1}, \mathbf{x_2})$. For a positive integer $i$, we denote by $[i]$ the set $\lbrace 1, 2, \cdots, i\rbrace$.  For any set of ordered indices $S=\lbrace s_1, s_2, \cdots, s_{|S|}\rbrace \subseteq \mathbb{Z}$, where $s_1< s_2 < \cdots < s_{|S|}$, and for any ensemble of variables $\lbrace X_i : i \in S\rbrace$, the tuple $(X_{s_1}, X_{s_2}, \cdots, X_{s_{|S|}})$ is denoted by $X_S$. We use $\log (.)$ to denote the  logarithm to the base $2$ and $H(.)$ to denote the binary entropy function. We use the notation $[2^K]$ to denote the set of $K$-length binary strings. A \emph{code} of length $n$ and dimension $k$ over alphabet $\mathcal{A}$ consists of an injective mapping $\mathcal{C}:\mathcal{A}^{k} \rightarrow \mathcal{A}^{n}$. When $\mathcal{A}$ is a finite field and the mapping $\mathcal{C}$ is linear, then the code is referred to as a \emph{linear code}. We refer to a linear code $\mathcal{C}$ of length $n$ and dimension $k$ as an $(n,k)$ code. An $(n,k)$ linear code is called MDS if the mapping projected to \emph{any} $k$ co-ordinates is invertible. 

\subsection{Multi-version Codes (MVCs)}
We now present a variant of the definition of the multi-version code \cite{multi_arxiv}, where we model the correlations. We consider a distributed storage system with $n$ servers that can tolerate $f$ crash-stop server failures. \color{black}   The system stores $\nu$ possibly correlated versions of a message where $\mathbf W_u \in [2^K]$ is the $u$-th version, $u \in [\nu]$, and $K$ is the message length in bits. The versions are assumed to be totally ordered, i.e., if $u>l$, $\mathbf W_u$ is interpreted as a \emph{later} version with respect to $\mathbf W_l.$ We assume that $\mathbf W_{1}\rightarrow \mathbf W_{2} \rightarrow \ldots \rightarrow \mathbf W_{\nu}$ form a Markov chain. $\mathbf W_{1}$ is uniformly distributed over the set of $K$ length binary vectors. Given $\mathbf W_{m}, \mathbf W_{m+1}$ is uniformly distributed in a Hamming ball of radius $\delta_K K$, $B(\mathbf W_m, \delta_K K)=\{ \mathbf W: d_H(\mathbf W, \mathbf W_m) \leq \delta_K K \}$, and the volume of the Hamming ball is given by
\begin{align}
Vol(\delta_K K, K)=|B(\mathbf W_m, \delta_K K)| = \sum\limits_{j=0}^{\delta_K K}  \binom{K}{j}.
\end{align}
Given a correlation coefficient $\delta_K$, we denote the set of possible tuples $(\mathbf w_1, \mathbf w_2, \cdots, \mathbf w_{\nu})$ under our correlation model by $A_{\delta_K}$. We provide the formal definition next.  
\begin{definition}[$\delta_K$-possible Set of Tuples] 
	\label{def:possible-set}
	The set $A_{\delta_K}$ of $\delta_K$-possible set of tuples \\ $(\mathbf w_{1}, \mathbf w_{2}, \cdots, \mathbf w_{\nu})$ is defined as follows 
	\begin{align}
	&A_{\delta_K} (\mathbf W_{1}, \mathbf W_{2}, \cdots, \mathbf W_{\nu})= \{  (\mathbf w_{1}, \mathbf w_{2}, \cdots, \mathbf w_{\nu}): \mathbf w_{1} \in [2^K], \mathbf w_{2} \in B(\mathbf w_1, \delta_K K),   \\ &\mathbf w_{3} \in B(\mathbf w_{2}, \delta_K K), \notag \cdots, \mathbf w_{\nu} \in B(\mathbf w_{\nu-1}, \delta_K K) \}.
	\end{align}
\end{definition}
\noindent We omit the dependency on the messages and simply write $A_{\delta_K}$, when it is clear from the context. Similarly, we can also define  the set of possible tuples $\mathbf w_{F_1}$ given a particular tuple $\mathbf w_{F_2}$, $A_{\delta_K}(\mathbf W_{F_1}|{\mathbf w_{F_2}})$, where $F_1, F_2$ are two subsets of $[\nu]$. 
\begin{remark}
Unlike the case of the twin binary symmetric source, in our model, the correlation coefficient $\delta_K$ is a function of $K$ in general and is not necessarily a constant. The more familiar expressions that involve entropies can be obtained when $\delta_K$ is equal to a constant $\delta$ using Stirling's inequality \cite{cover_thomas}. Specifically, for $\delta<1/2$, we have
\begin{align}
KH(\delta)-o(K) \leq \log Vol(\delta K, K) \leq K H(\delta).
\end{align} 
\end{remark}
\color{black}
The $i$-th server receives an arbitrary subset of versions $\mathbf{S}(i) \subseteq [\nu]$ that denotes the \emph{state} of that server. We denote the system state by $\mathbf{S}= \lbrace \mathbf{S}(1), \mathbf{S}(2), \cdots, \mathbf{S}(n)\rbrace \in \mathcal{P}([\nu])^n$, where $\mathcal{P}([\nu])$ is the power set of $[\nu]$. For the $i$-th server with state $\mathbf S(i)=\lbrace s_1, s_2, \cdots, s_{|\mathbf S(i)|} \rbrace$, where $s_1 < s_2 < \cdots <s_{|\mathbf S(i)|}$, the server stores a codeword symbol generated by the encoding function $\varphi_{\mathbf S(i)}^{(i)}$ that takes an input $\mathbf W_{\mathbf S(i)}$ and outputs an element from the set $[q]$ that can be represented by $\log q$ bits. In state $\mathbf S \in \mathcal{P}([\nu])^n$, we denote the set of servers that have received $\mathbf W_{u}$ by $\mathcal A_{\mathbf S}(u)=\{j \in [n]:  u \in \mathbf{S}(j)\}$ and a version $u \in [\nu]$ is termed \emph{complete} if $|\mathcal A_{\mathbf S}(u)| \geq c_W$, where $c_W \leq n-f$.
The set of complete versions in state $\mathbf S \in \mathcal P([\nu])^n$ is given by $\mathcal C_{\mathbf S} = \{u \in [\nu]: |\mathcal A_{\mathbf S}(u)| \geq c_W\}$ and the latest among them is denoted by  $L_{\mathbf S}\coloneqq\max \  \mathcal C_{\mathbf S}$.

 The goal of the multi-version coding problem is to devise encoders such that for every decoder that connects to any arbitrary set of $c_R \leq n-f$ servers, the latest complete version or a later version is decodable with probability of error that is at most $\epsilon$ while minimizing the per-server worst-case storage cost. We express this formally next. 
 \begin{definition}[$\epsilon$-error $(n,c_W, c_R, \nu, 2^K, q, \delta_K)$ multi-version code (MVC)]	\label{Multi-version Code Definition} An $\epsilon$-error \\ $(n,c_W, c_R, \nu, 2^K, q, \delta_K)$ multi-version code (MVC) consists of the following 
	\begin{itemize}
		\item encoding functions
	$	
		\varphi_{\mathbf S(i)}^{(i)} \colon {[2^K]}^{|\mathbf S(i)|} \to [q],\ \textrm{for every}\ i \in [n]\ \textrm{and every state}\  \mathbf S(i) \subseteq [\nu]$
		\item decoding functions $\psi_{\mathbf{S}}^{(T)} \colon [q]^{c_R} \to [2^K] \cup \lbrace \textit{NULL} \rbrace,$
\end{itemize}
that satisfy the following
	\begin{align*}
		\Pr \left[ \psi_{\mathbf S}^{(T)} \left( \varphi_{\mathbf S(t_1)}^{(t_1)}, \cdots, \varphi_{\mathbf S(t_{c_R})}^{(t_{c_R})} \right)  
		=\mathbf W_m   \ \text{for some} \ m \geq L_{\mathbf S},  \text{if} \ \mathcal C_{\mathbf S} \neq \emptyset \right] \geq 1-\epsilon,
		\end{align*}
for every possible system state $\mathbf{S} \in \mathcal{P}([\nu])^n$ and every set of servers $T=\lbrace t_1, t_2, \cdots, t_{c_R} \rbrace \subseteq [n]$, where the probability is computed over all possible tuples of the message versions.
 \end{definition}

We notice that set of possible tuples $A_{\delta_K}$ can be partitioned into disjoint sets as follows
\begin{align}
	A_{\delta_K}= A_{\delta_K, 1} \cup A_{\delta_K, 2},
\end{align}
where $A_{\delta_K, 1}$ is the set of tuples $(\mathbf{w}_1, \mathbf{w}_2, \cdots, \mathbf w_\nu) \in A_{\delta_K}$ for which we can decode successfully for all $\mathbf S  \in \mathcal{P}([\nu])^n$ and $A_{\delta_K, 2}$ is the set of tuples where we cannot decode successfully at least for one state $\mathbf S  \in \mathcal{P}([\nu])^n$. \color{black}
\begin{definition}[Storage Cost of a Multi-version Code]
	The storage cost of an $\epsilon$-error \\ $(n, c_W, c_R, \nu, 2^K, q, \delta_K)$ MVC is equal
	to $\alpha=\log q$ bits.
\end{definition}
\color{blue}
\color{black}
	We next present an alternative decoding requirement that is shown in \cite{multi_arxiv} to be equivalent to the multi-version coding problem defined above. For any set of servers $T \subseteq [n]$, note that $\max \cap_{i \in T} \mathbf S(i)$ denotes \emph{the latest common version among} these servers. The alternate decoding requirement, which we refer to multi-version coding problem with Decoding Requirement A, replaces $c_W, c_R$ by one parameter $c$. The decoding requirement requires that the decoder connects to any $c$ servers and decodes the \emph{latest common version} amongst those $c$ servers, or a later version.
\begin{definition} [$\epsilon$-error  $(n, c, \nu, 2^K, q, \delta_K)$ Multi-version code (MVC) with Decoding Requirement A] 
An $\epsilon$-error $(n,c, \nu, 2^K, q, \delta_K)$ multi-version code (MVC) consists of the following
	\begin{itemize}
		\item encoding functions
		$\varphi_{\mathbf S(i)}^{(i)} \colon {[2^K]}^{|\mathbf S(i)|} \to [q],\ \textrm{for every}\ i \in [n]\ \textrm{and every state}\  \mathbf S(i) \subseteq [\nu]$
		\item decoding functions
		$\psi_{\mathbf{S}}^{(T)} \colon [q]^{c} \to [2^K] \cup \lbrace \textit{NULL} \rbrace,$
\end{itemize}
that satisfy the following
    \begin{align*}
     \Pr \left[ \psi_{\mathbf{S}}^{(T)} \left( \varphi_{\mathbf S(t_1)}^{(t_1)}, \cdots, \varphi_{\mathbf S(t_c)}^{(t_c)} \right) = \mathbf W_m, \text{for some} \ m \geq \max \cap_{i \in T} \mathbf S(i), \ \text{if} \cap_{i \in T} \mathbf S(i) \neq \emptyset \right] \geq 1-\epsilon,
	\end{align*}
for every possible system state $\mathbf{S} \in \mathcal{P}([\nu])^n$ and every set of servers $T=\lbrace t_1, t_2, \cdots, t_{c} \rbrace \subseteq [n]$, where the probability is computed over all possible tuples of the message versions.
\end{definition} 
\color{black}
\noindent In this paper, we present our achievability results for decoding requirement A and Lemma \ref{achievability lemma} \cite{multi_arxiv} establishes the connection between the two decoding requirements.

\begin{lemma}
\label{achievability lemma}	
Consider any positive integers $n,c_W,c_R,c$ such that $c=c_W+c_R-n$.
An $\epsilon$-error $(n, c, \nu, 2^K, q, \delta_K)$ MVC with decoding requirement A exists if and only if  an $\epsilon$-error $(n, c_W, c_R, \nu, 2^K, q, \delta_K)$ MVC exists. 
\end{lemma}
\remove{Although the mathematical definition of multi-version codes expects that all the versions are available at the server for encoding simultaneously, in practical systems, the versions arrive at the server one after another. 
Motivated by the requirement to capture scenarios where versions arrive one after another at the servers, reference \cite{multi_arxiv} defined the notion of \emph{causal} multi-version codes. Informally, the encoding function of a causal multi-version code  has the following property: for any state $S$ and any $j \in S$, the symbol stored by a server in state $S$ can be expressed as a function of a codeword symbol created using versions in $S-\{j\}$ and the $j$-th version. We anticipate that causal multi-version codes are more amenable to applications in practical systems, as the encoding does not depend on the order of their arrival.  We next provide the formal definition. 
\begin{definition}\textbf{(Causal Codes)}
\label{Causal Codes Definition}
A multi-version code is called causal if it satisfies:
for all $S \subseteq [\nu], j \in S, i \in [n]$, there exists a function
\begin{align*}
 \hat{\varphi}_{S, j}^{(i)}: [q] \times [2^K] \to [q], 
\end{align*}
such that
\begin{align*}
& \varphi_S^{(i)}(\mathbf W_S)= \hat{\varphi}_{S, j}^{(i)} (\varphi_{S \setminus \lbrace j \rbrace}^{(i)} (\mathbf W_{S \setminus \lbrace j \rbrace}), \mathbf W_j).
\end{align*}
\end{definition}}

We now make some remarks interpreting the multi-version coding problem in terms of the underlying system and algorithm that it aims to model.
\\
\begin{remark} \begin{enumerate}
\item[]		
\item[A.] The system model has an implicit failure tolerance of $f$, as along as the quorum sizes $c_W, c_R$ are chosen such that $c_W, c_R \leq n-f$. This is because, for a version to be complete, a writer must contact $c_W$ severs, and for a reader, it must obtain responses from $c_R$ servers - choosing $c_W, c_R \leq n-f$ ensures that write and read operations complete provided that the number of failed servers is no larger than $f$ (see \cite{Lynch1996, ABD} for more details).
\item[B.] The parameter $\nu$ can be interpreted as a measure of the number of concurrent writes in the system \cite{multi_arxiv,Cadambe_Wang_Lynch2016,CadambeCoded_NCA, Dutta}. In distributed algorithms, the ordering among the various write operations is determined by carefully constructed protocols usually through the use of \emph{Lamport timestamps} (also known as logical timestamps) \cite{lamport1978time}. For instance, several protocols (e.g., \cite{ABD, Dutta, CadambeCoded_NCA}) involve a separate round of communication for a write to figure out the lamport clock (i.e., version number) before proceeding with dispersing the data. The mult-version coding problem abstracts out this these protocol details into the version number. However, it is worth noting that a ``later version'' is not necessarily arriving to the system after an earlier version - they can be concurrent, and can in fact arrive at different nodes at different orders (e.g., see \cite{multi_arxiv, Cadambe_Wang_Lynch2016}).  A later version may simply be viewed as one that could receive a higher Lamport timestamp in a protocol execution.
	\item[C.] Unlike the study of \cite{multi_arxiv} which considers $0$-error MVCs, we allow the probability of error to be at most $\epsilon$. See also Remark \ref{rem:zeroerror} for more details. 
		\end{enumerate}
\end{remark}
\color{black}
\subsection{Background - Replication and Simple Erasure Coding}
\label{sec:background}
Replication and simple MDS codes provide two natural MVC constructions. Suppose that the state of the $i$-th server is $\mathbf S(i)=\{s_1, s_2, \ldots, s_{|\mathbf S(i)|}\}$, where $s_1 < s_2< \ldots< s_{|\mathbf S(i)|}$. 
\begin{itemize}	
	\item \emph{Replication-based MVCs:} In this scheme, each server only stores the latest version it receives. The encoding function is $\varphi_{\mathbf S(i)}^{(i)}(\mathbf W_{\mathbf S(i)}) = \mathbf W_{s_{|\mathbf S(i)|}}$, hence the storage cost is $K$.
	
	\item \emph{Simple MDS codes based MVC (MDS-MVC):} In this scheme, an $(n, c)$ MDS code is used to encode each version separately. Specifically, suppose that $\mathcal{C}:[2^K]\rightarrow [2^{K/c}]^{n}$ is an $(n,c)$ MDS code over alphabet $[2^{K/c}],$ and denote the $i$-th co-ordinate of the output of $\mathcal{C}$ by $\mathcal{C}^{(i)}:[2^K]\rightarrow[2^{K/c}]$. The encoding function is constructed as $\varphi_{\mathbf S(i)}^{(i)}(\mathbf W_{\mathbf S(i)}) = (\mathcal{C}^{(i)}(\mathbf W_{s_1}), \mathcal{C}^{(i)}(\mathbf W_{s_2}),\ldots, \mathcal{C}^{(i)}(\mathbf W_{s_{|\mathbf S(i)|}}))$. That is, each server stores one codeword symbol for each version it receives and the storage cost is $\nu \frac{K}{c}$. 
\end{itemize}
An important outcome of the study of \cite{multi_arxiv} is that, when the different versions are independent, i.e., if $\delta_K = 1$, then the storage cost is at least $\frac{\nu K}{\nu+c-1}-\Theta(1)$. In particular, because $\frac{\nu}{\nu+c-1} \geq \frac{1}{2} \min(\frac{\nu}{c}, 1)$, the best possible MVC scheme is, for large $K$, at most twice as cost-efficient as the better among replication and simple erasure coding. In this paper, we show that replication and simple erasure coding are significantly inefficient if the different versions are correlated. Our schemes resemble simple erasure codes in their construction; however, we exploit the correlation between the versions to store fewer bits per server. 

\subsection{Summary of Results}
\label{Main Results}
\label{sec:motivation_results}
In order to explain the significance of our results, summarized in Table \ref{table:Schemes Comparison}, we begin with a simple motivating scheme. Consider the MDS-MVC scheme of Section \ref{sec:background}. Assume that we use a Reed-Solomon code over a field $\mathbb{F}_{p}$ of binary characteristic. The generator matrix of a Reed-Solomon code is usually expressed over $\mathbb{F}_{p}$. However, every element in $\mathbb{F}_{p}$ is a vector over $\mathbb{F}_{2}$, and a multiplication over the extension field $\mathbb{F}_{p}$ is a linear transformation over $\mathbb{F}_{2}.$ Therefore, the generator matrix of the Reed-Solomon code can be equivalently expressed over $\mathbb{F}_{2}$ as follows

$$ G = (G^{(1)}, G^{(2)}, \ldots, G^{(n)}),$$
where $G$ is a $K \times nK/c$ binary generator matrix, and $G^{(i)}$ has dimension $K \times K/c$. Because Reed-Solomon codes can tolerate $n-c$ erasures, every matrix of the form $(G^{(t_1)}, G^{(t_2)}, \ldots, G^{(t_c)}),$ where $t_1, t_2,\ldots, t_c$ are distinct elements of $[n]$, has a full rank of $K$ over $\mathbb{F}_2$. \\
We now describe a simple scheme that extends the MDS-MVC by exploiting the correlations and requires the knowledge of $\delta_{K}$. Suppose that the $i$-th server receives the set of versions $\mathbf S(i)=\lbrace s_1, s_2, \cdots, s_{|\mathbf S(i)|}\rbrace$, where $s_1 < s_2< \ldots< s_{|\mathbf S(i)|}$. The server encodes $\mathbf W_{s_{1}}$ using the binary code as $\mathbf W_{s_1} ^{\rm T}  G^{(i)}$. For $\mathbf W_{s_{m}}$, where $m>1$, the server finds a difference vector $\mathbf y^{(i)}_{s_m, s_{m-1}}$ that satisfies the following
\begin{enumerate}
	\item  $\mathbf y^{(i) \ \rm T}_{s_m, s_{m-1}} G^{(i)} = (\mathbf W_{s_{m}}- \mathbf W_{s_{m-1}})^{\rm T} G^{(i)}$ and 
\item $w_H(\mathbf y^{(i)}_{s_m, s_{m-1}}) \leq (s_m-s_{m-1}) \delta_K K.$   
\end{enumerate}
Although it is not necessary that $\mathbf y^{(i)}_{s_m, s_{m-1}} = \mathbf W_{s_{m}}-\mathbf W_{s_{m-1}}$, the fact that $\mathbf W_{s_{m}}-\mathbf W_{s_{m-1}}$ satisfies these  two conditions implies that the encoder can find at least one vector $\mathbf{y}^{(i)}_{s_m,s_{m-1}}$ satisfying these conditions. Since $w_H(\mathbf y^{(i)}_{s_m, s_{m-1}}) \leq (s_m-s_{m-1}) \delta_K K$, an encoder aware of $\delta_{K}$ can represent $\mathbf y^{(i)}_{s_m, s_{m-1}}$ by $\log Vol((s_m-s_{m-1})\delta_K K, K)$ bits. The first condition implies that a decoder that connects to the $i$-th server can obtain $\mathbf W_{s_m}^{\rm T}G^{(i)}$ by applying $\mathbf W_{s_1}^{\rm T}G^{(i)}+ \sum\limits_{\ell=2}^{m}\mathbf y^{(i) \ \rm T}_{s_\ell, s_{\ell-1}} G^{(i)}  =  \mathbf W_{s_m}^{\rm T}G^{(i)}$. Therefore, from any subset $\{t_1,t_2,\ldots,t_{c}\}$ of $c$ servers, for any common version $s_{m}$ among these servers, a decoder can recover $\mathbf W_{s_m}^{\rm T}G^{(t_{1})}, \mathbf W_{s_m}^{\rm T}G^{(t_{2})},\ldots,\mathbf W_{s_m}^{\rm T}G^{(t_{c})}$ from these servers and can therefore recover $\mathbf W_{s_{m}}$. 

The worst-case storage cost of this scheme is obtained when each server receives all the $\nu$ versions, which results in a storage cost of $\frac{K}{c}+ (\nu-1) \log  Vol(\delta_K K, K).$ Intuitively, the above scheme stores the first version using erasure coding - $K/c$ bits - and the remaining $(\nu-1)$ versions using delta coding, which adds a storage cost of $\log Vol(\delta_K K, K)$ bits per version.
\begin{table*}[t]
	\renewcommand{\arraystretch}{1.2}
	\centering
	\begin{tabular}{ | p{3.7 cm} | p{7 cm} | p{4cm} | }
		\hline
		\textbf{Scheme} & \textbf{Worst-case Storage Cost} & \textbf{Regime} \\ \hline
		Replication  &  $K$ & oblivious to $\delta_K$, $\epsilon=0$
		\\ \hline	
		Simple erasure codes &  $\nu \frac{K}{c}$& outperforms replication if $\nu<c$, oblivious to $\delta_K$, $\epsilon=0$  \\ \hline
		Theorem \ref{Reed-Solomon Update-Efficient MVC Theorem} [Reed-Solomon update-efficient code] & $\frac{K}{c}+ (\nu-1) \delta_K K (\log K + o(\log K))$& asymptotically outperforms the above schemes for $\nu<c$ and $\delta_K=o(1/\log K)$, oblivious to $\delta_K$, $\epsilon=0$ \\\hline
		Motivating scheme of Section \ref{sec:motivation_results} &$
		\frac{K}{c}+(\nu -1) \log Vol(\delta_K K, K)$& not oblivious to $\delta_K$, $\epsilon=0$\\ \hline
		Theorem \ref{Slepian-Wolf Storage Cost Theorem} [Random binning] &$
		\frac{K}{c}+\frac{\nu -1}{c} {\log Vol(\delta_K K, K)} + o(\log K)$& asymptotically outperforms the above schemes for $\nu < c$, not oblivious to $\delta_K$, $\epsilon=1/\log K$ \\ \hline
		Theorem \ref{Lower Bound Theorem} [Lower bound]
		& $\frac{K}{c+\nu-1}+\frac{\nu-1}{c+\nu-1} \log Vol (K, \delta_K K) - \Theta(1)$ 
		& applicable for all $\delta_K$
		\\ \hline 
	\end{tabular}
	\caption{Storage cost.}\label{table:Schemes Comparison}
\end{table*}
The scheme we described above motivates the following two questions.\\
\emph{\textbf{Q1}: Can we obtain a MVC construction that is oblivious to the parameter $\delta_K$ with a storage cost of $\frac{K}{c}+ (\nu-1) \log  Vol(\delta_K K, K)$?}\\
\emph{\textbf{Q2}: Can we use erasure coding to achieve a storage cost of $\frac{K}{c}+ \frac{\nu-1}{c} \log  Vol(\delta_K K, K)$?}

\noindent In Section \ref{Update-efficient Multi-version Codes}, we provide Theorem \ref{Reed-Solomon Update-Efficient MVC Theorem} that answers $Q1$ by developing a $0$-error Reed-Solomon based scheme that does not require the knowledge of $\delta_K$ and obtains the erasure coding gain of $1/c$ for the  first  version  available  at  a  server and  stores  the  subsequent  versions  via delta coding. 

In Section \ref{Slepian-Wolf Based Multi-version Codes}, we provide 
Theorem \ref{Slepian-Wolf Storage Cost Theorem} that gives a positive answer to $Q2$ by showing the existence of an $\epsilon$-error storage efficient scheme that obtains the erasure coding factor of $1/c$, not only for the first version, but also for the subsequent versions. Moreover, the scheme is able to harness the delta compression gain. 

Finally, in Section \ref{Lower Bound on the storage cost}, 
Theorem \ref{Lower Bound Theorem} provides a lower bound on the per-server storage cost which implies that for $\nu < c$, constant $\delta_K=\delta$ and $\epsilon=2^{-o(K)}$, the achievable scheme of Theorem \ref{Slepian-Wolf Storage Cost Theorem} is asymptotically at most twice the lower bound. We notice that the regime where $\nu<c$ is interesting as the degree of asynchrony is typically limited as pointed out in \cite{chen2017giza}. \color{black}



\section{Code Constructions (Theorem \ref{Reed-Solomon Update-Efficient MVC Theorem} and Theorem \ref{Slepian-Wolf Storage Cost Theorem}) }
\allowdisplaybreaks{
In this section, we provide our code constructions. We study the case where $\delta_K$ is not known and present a MVC based on Reed-Solomon code in Section \ref{Update-efficient Multi-version Codes}. Later on in this section, we study the case where $\delta_K$ is known and propose a random binning argument in Section \ref{Slepian-Wolf Based Multi-version Codes}. 
\label{Code Constructions}
\subsection{Update-efficient Multi-version Codes}
\label{Update-efficient Multi-version Codes}
 We develop simple multi-version coding scheme that exploits the correlation between the different versions and have smaller storage cost as compared with \cite{multi_arxiv}. In this scheme, the servers do not know the correlation degree $\delta_K$ in advance. We begin by recalling the definition of the update efficiency of a code from \cite{anthapadmanabhan2010update}. \begin{definition}[Update efficiency]
 	For a code $\mathcal {C}$ of length $N$ and dimension $K$ with encoder $\mathcal{C}:\ \mathbb{F}^{K} \rightarrow \  \mathbb{F}^{N},$ the update efficiency of the code is the maximum number of codeword symbols that must be updated when a single message symbol is changed and is expressed as follows
 	\begin{align}
 	t=\max_{\substack{\mathbf W, \mathbf W' \in \mathbb{F}^K:\\ d_H(\mathbf W, \mathbf W')=1}} d_H(\mathcal{C}(\mathbf W), \mathcal{C}(\mathbf W')).
 	\end{align}
 \end{definition}
   \noindent An $( N, K)$ code $\mathcal{C}$ is referred to as update-efficient code if it has an update efficiency of $o( N)$. 
 \begin{definition}[Update efficiency of a server]
 	Suppose that $\mathcal{C}^{(i)}:\mathbb{F}^{K} \rightarrow \  \mathbb{F}^{N/n}$ denotes the $i$-th co-ordinate of the output of $\mathcal{C}$ stored by the $i$-th server. The update efficiency of the $i$-th server is the maximum number of codeword symbols that must be updated in this server when a single message symbol is changed and is expressed as follows
 	\begin{align}
 	t^{(i)}=\max_{\substack{\mathbf W, \mathbf W' \in \mathbb{F}^K: \\ d_H(\mathbf W, \mathbf W')=1}} d_H (\mathcal{C}^{(i)}(\mathbf W), \mathcal{C}^{(i)}(\mathbf W')).
 	\end{align}
 \end{definition} \noindent 
 Suppose that $G=(G^{(1)}, G^{(2)}, \cdots, G^{(n)})$ is the generator matrix of a linear code $\mathcal{C}$, where $G^{(i)}$ is a $K \times N/n$ matrix that corresponds to the $i$-th server. The update efficiency of the $i$-th server is the maximum row weight of $G^{(i)}$. \color{black}
 \begin{definition}[The per-server maximum update efficiency]
 	The per-server maximum update efficiency is the maximum number of codeword symbols that must be updated in any server when a single message symbol is changed and is given by
 	\begin{equation}
 	t_s=\max_{\substack{i \in [n]}} \ t^{(i)}.
 	\end{equation}
 \end{definition}
We next present an update-efficient MVC construction, illustrated in Fig. \ref{Reed-Solomon}, that is based on Reed-Solomon code and has a maximum update efficiency per-server $t_s=1$. 
\begin{construction}[Reed-Solomon Update-Efficient MVC] 
	\label{update-efficient construction} Suppose that the $i$-th server receives the versions $\mathbf S(i)=\{s_1, s_2, \ldots, s_{|\mathbf S(i)|}\}$, where $s_1<s_2<\cdots<s_{|\mathbf S(i)|}$. A version $\mathbf W_{s_j}$ is divided into $\frac{K}{c \log n_p}$ blocks of length $ c \log n_p,$ where $n_p=2^{\lceil \log_{2}n\rceil}$. In each block, every consecutive string of $\log n_p$ bits is represented by a symbol in $\mathbb{F}_{n_p}$. The representation of  $\mathbf W_{s_j}$ over $\mathbb{F}_{n_p}$  is denoted by $\mathbf {\overline W}_{s_j}$. Each block is then encoded by an $(n, c)$ Reed-Solomon code with a generator matrix $\tilde G$ that is given by  
	\begin{align}
	\tilde G &=\left( 
	\begin{array}{c c c c}
	1 & 1 & \cdots  & 1 \\
	\lambda_1  & \lambda_2  & \cdots &   \lambda_n \\
	\lambda_1^2  & \lambda_2^2  & \cdots &   \lambda_n^2 \\
	\vdots & \vdots & \vdots & \vdots \\
	\lambda_1^{c-1}&  \lambda_2^{c-1} & \cdots & \lambda_n^{c-1} 
	\end{array}
	\right),
	\end{align}
	where $\Lambda=\left\lbrace  \lambda_1, \lambda_2, \cdots, \lambda_n \right\rbrace  \subset \mathbb{F}_{n_p}$ is a set of distinct elements. 
	For $\mathbf W_{s_1}$, the $i$-th server stores $\mathbf {\overline W}_{s_1}^{\rm T} G^{(i)}$, where $G^{(i)}$ is a $ \frac{K}{\log n_p} \times \frac{K}{c \log n_p}$ matrix that is given by
	\begin{align}
	G ^{(i)}&=\left( 
	\begin{array}{c c c c c}
	\tilde G \mathbf e_i & 0 & \cdots &0 & 0 \\
	0  & \tilde G \mathbf e_i   & \cdots &  0 &  0 \\
	\vdots & \vdots & \vdots & \vdots & \vdots\\
	0 & 0 & \cdots &0 & \tilde G  \mathbf e_i 
	\end{array}
	\right),
	\end{align} 
	where $\mathbf{e}_{i}$ is $i$-th standard basis vector over $\mathbb F_{n_p}$. For $\mathbf{W}_{s_m}$, where $m>1$, the server may only store the updated symbols from the old version $\mathbf{\overline{W}}_{s_{m-1}}$ or store  $\mathbf {\overline W}_{s_m}^{\rm T} G^{(i)}$. 
\end{construction}
\begin{figure}[h]
	\centering
	\includegraphics[width=.7\textwidth,height=.22\textheight]{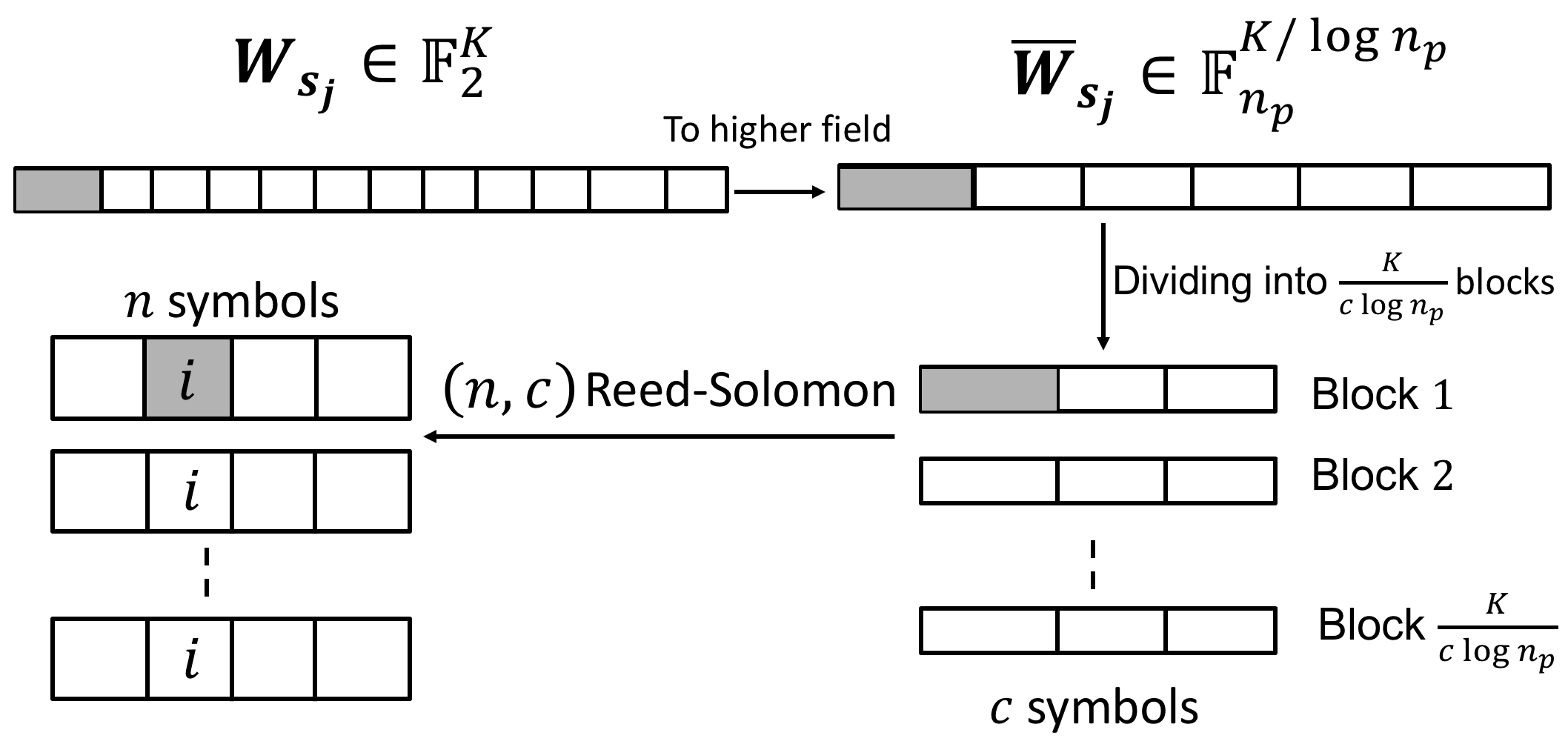}
	\caption{Illustration of Construction \ref{update-efficient construction}. A bit that changes in the message leads at most to one updated codeword symbol in the $i$-th server, $\forall i \in [n]$, hence $t_s=1$.  \label{Reed-Solomon} \color{black}}
\end{figure}
\begin{theorem}
	\label{Reed-Solomon Update-Efficient MVC Theorem}[Reed-Solomon Update-Efficient MVC]
	Construction \ref{update-efficient construction} is a $0$-error $(n, c, \nu, 2^K, q, \delta_K)$ multi-version code with a worst-case storage that is at most 
	\begin{equation}
	\frac{K}{c}+   (\nu-1) \min(\delta_K K \log \left( \frac{Kn_p}{c\log n_p}\right), K/c),
	\end{equation}
	where $n_p=2^{\lceil \log_{2}n\rceil}$.
\end{theorem} 
\begin{proof} [Proof of Theorem \ref{Reed-Solomon Update-Efficient MVC Theorem}]
We observe that Construction \ref{update-efficient construction} is a valid multi-version code as the latest common version is recoverable from any $c$ servers by the MDS property of Reed-Solomon code. In order to characterize the worst-case storage cost, we observe that the update efficiency of the $i$-th server is equal to the maximum row weight of $G^{(i)}$ which is equal to $1, \forall i \in [n]$. Thus, the per-server maximum update efficiency $t_s$ is also equal to $1$. \\ The worst-case storage cost corresponds to the case where a server receives all versions. In this case, the server stores $\mathbf {\overline W}_{1}^{\rm T} G^{(i)}$ for the first version. For $\mathbf{W}_{u}$, where $u>2$, the server may only store the updated symbols from the old version $\mathbf{\overline{W}}_{{u-1}}$ or store $\mathbf {\overline W}_{u}^{\rm T} G^{(i)}$. Storing the index of an updated symbol requires $\log (\frac{K}{c\log n_p})$ bits and storing the value requires $\log n_p$. Therefore, the per-server storage cost is upper-bounded as follows
	\begin{align*}
	\alpha &\leq \frac{K}{c}+ \sum_{u=2}^{\nu} \min(d_H(\mathbf W_{u}, \mathbf W_{u-1}) \log \left( \frac{Kn_p}{c\log n_p}\right), K/c) \\
	& \leq  \frac{K}{c}+ \sum_{u=2}^{\nu} \min(d_H(\mathbf W_{u}, \mathbf W_{u-1}) \log \left( \frac{Kn_p}{c\log n_p}\right), K/c)	\\
	& \leq \frac{K}{c}+   (\nu-1) \min(\delta_K K \log \left( \frac{Kn_p}{c\log n_p}\right), K/c).
	\end{align*}
\end{proof}
\color{black}
\subsection{Random Binning Based Multi-version Codes}
\label{Slepian-Wolf Based Multi-version Codes}
We next introduce a random binning argument showing the existence of a multi-version code that can harness both of the erasure and the delta coding gains for all versions for the case where $\delta_K$ is known. Recall that Slepian-Wolf coding \cite{slepian1973noiseless,wolf1973data} is a distributed data compression technique for correlated sources that are drawn in independent and identical manner according to a given distribution. In the Slepian-Wolf setting, the decoder is interested in decoding the data of all sources. In the multi-version coding problem, the decoder is interested in decoding the latest common version, or a later version, among any set of $c$ servers.

We notice that our model differs from the Slepian-Wolf setting as we do not aspire to decode all the versions.
The lossless source coding problem with a helper \cite{ahlswede1975source, wyner1975source, cover1975proof} may also seem related to our approach, since the side information of the older versions may be interpreted as helpers. In the optimal strategy for the helper setting, the helper side information is encoded via a joint typicality encoding scheme, whereas the random binning is used for the message. However, in the multi-version coding setting, a version that may be a side information for one state may be required to be decoded in another state. For this reason, a random binning scheme for all versions leads to schemes with a near-optimal storage cost.  We next present the code construction that is inspired by Cover's random binning proof of the Slepian-Wolf problem \cite{cover1975proof}. 
\begin{construction}[Random binning multi-version code] 
\label{Random binning multi-version code}	
	Suppose that the $i$-th server receives the versions $\mathbf S(i)=\{s_1, s_2, \cdots, s_{|\mathbf S(i)|}\} \subseteq [\nu]$, where $s_1<s_2<\cdots<s_{|\mathbf S(i)|}$.
	\begin{itemize}
		\item  \emph{Random code generation:} At the $i$-th server, for a version $s_j$ the encoder assigns an index at random from $\lbrace 1, 2,  \cdots, 2^{R_{s_j}^{(i)}/c} \rbrace$ uniformly and independently to each vector of length $K$ bits, where  $R_{s_j}^{(i)}/c$ is the rate assigned by the $i$-th server to version $s_j$. 
		\item \emph{Encoding}:\label{Slepian-Wolf Encoder} The server stores the corresponding index to each version that it receives and the decoder is also aware of this mapping. The encoding function of the $i$-th server is given by
		\begin{align}
		\varphi_{\mathbf S(i)}^{(i)}=( \varphi_{ s_1 }^{(i)},  \varphi_{ s_2 }^{(i)}, \cdots,  \varphi_{s_{|\mathbf S(i)|}}^{(i)} ),
		\end{align}
		where $\varphi_{ s_j }^{(i)} \colon [2^K]  \to \lbrace 1, 2 \cdots, 2^{K R_{s_j}^{(i)}/c} \rbrace$ 
		and we choose the rates as follows
		\begin{align}
		K R_{s_1}^{(i)}&=K+(s_1-1)\log Vol(\delta_K K, K)+ (s_1-1)-\log \epsilon 2^{- \nu n},\label{eq:SWrate1}\\
		KR_{s_j}^{(i)}&=(s_j-s_{j-1}) \log Vol(\delta_K K, K)+(s_j-1)-\log \epsilon 2^{-\nu n},   j \in \lbrace 2, 3, \cdots, |\mathbf S(i)|\rbrace. \label{eq:SWrate2}
		\end{align}
\end{itemize}	
		Consider a state $\mathbf{S} \in \mathcal{P}([\nu])^n$ and suppose that the decoder connects to the servers $T=\{t_1, t_2, \cdots, t_c\} \subseteq [n]$. Suppose that a version $s_j$ is received by a set of servers $\{i_1, i_2, \cdots, i_r \} \subseteq T$, then the bin index corresponding to this version is given by 
		\begin{align}
		\varphi_{ s_j }= (\varphi_{ s_j }^{(i_1)}, \varphi_{ s_j }^{(i_2)}, \cdots, \varphi_{ s_j }^{(i_r)}).
		\end{align} 
		In this case, the rate of version $s_j$ is given by
		\begin{align}
		\label{sum_rate_all_servers}
		R_{s_j}=\frac{1}{c}\sum_{i \in \{i_1, i_2, \cdots, i_r \} } R_{s_j}^{(i)}.
		\end{align}
\item \emph{Decoding}:\label{Slepian-Wolf Decoder} 
The decoder employs the \emph{possible set decoding} strategy as follows. Assume that $\mathbf W_{u_L}$ is the latest common version in $\mathbf{S}$ and that the versions $\mathbf W_{u_1}, \mathbf W_{u_2}, \cdots, \mathbf W_{u_{L-1}}$ are the older versions such that each of them is received by at least one server out of those $c$ servers. We denote this set of versions by $\mathbf S_{\rm T}$ and define it formally as follows 
\begin{align}
\label{set of versions}
\mathbf S_{\rm T} =\{u_1, u_2, \cdots, u_L\}=\left( \bigcup_{t \in T} \mathbf{S}(t)\right)  \setminus \{u_L+1, u_L+2, \cdots, \nu \},
\end{align}
where $u_1<u_2< \cdots<u_L$. Given the bin indices $(b_{u_1}, b_{u_2}, \cdots, b_{u_L})$, the decoder finds all tuples $(\mathbf w_{u_1}, \mathbf w_{u_2}, \cdots, \mathbf w_{u_L}) \in A_{\delta_K}$ such that $(\varphi_{u_1} (\mathbf w_{u_1})=b_{u_1}, \varphi_{u_2} (\mathbf w_{u_2})=b_{u_2}, \cdots, \varphi_{u_L} (\mathbf w_{u_L})=b_{u_L})$. If all of these tuples have the same latest common version $\mathbf w_{u_L}$, the decoder declares $\mathbf w_{u_L}$ to be the estimate of the latest common version $\mathbf{\hat W_{u_L}}$. Otherwise, it declares an error. 		
\end{construction}
\begin{theorem}
	\label{Slepian-Wolf Storage Cost Theorem}[Random Binning MVC]
	There exists an $\epsilon$-error $(n, c, \nu, 2^K, q, \delta_K)$ multi-version code whose worst-case storage cost is at most 
	\begin{align}
	\frac{K}{c} + \frac{(\nu -1) \log Vol(\delta_K K, K)}{c}+\frac{\nu(\nu-1)/2- \nu \log \epsilon 2^{-\nu n}}{c}.
	\end{align}
\end{theorem}
\begin{proof}[Proof of Theorem \ref{Slepian-Wolf Storage Cost Theorem}]
We show that Construction \ref{Random binning multi-version code} is an $\epsilon$-error multi-version code. \\ We denote the error event by $E$ and express it as follows 
    \begin{align}
    \label{binning error event}
      E = \{ \exists (\mathbf{w}'_{u_1}, \mathbf{w}'_{u_2}, \ldots,\mathbf{w}_{u_L}') \in  A_{\delta_K}: \mathbf{w}_{u_L}' \neq \mathbf W_{u_L} \ \textrm{and} ~~\varphi_u(\mathbf w'_{u})=\varphi_u(\mathbf W_{u}), \forall u \in \mathbf S_T\}.
	\end{align}
	The error event in decoding can be equivalently expressed as follows 
	\begin{align}
	E= \bigcup_{ \mathcal I \subseteq \mathbf {S}_T: u_L \in \mathcal{I}} E_{ \mathcal I}
	\end{align}
	where
	\begin{align}
	\label{Error Events}
	E_{\mathcal I} = &\{\exists \mathbf w'_u \neq \mathbf W_u, \forall u \in {\mathcal I}: \varphi_u(\mathbf w'_{u} )= \varphi_u(\mathbf W_{u}), \forall u \in {\mathcal I}  \ \text{and} \ (\mathbf w'_{\mathcal I}, \mathbf W_{\mathbf S_T \setminus \mathcal I}) \in A_{\delta_K}\}, 	
	\end{align}
	for $ \mathcal{I} \subseteq \mathbf S_T$ such that $u_L \in \mathcal I$. By the union bound, we have
	\begin{equation}
	\begin{aligned}
	P_e (\mathbf S, T)& \coloneqq P(E)  =P \left( \bigcup_{ \mathcal I \subseteq \mathbf {S}_T: u_L \in \mathcal I} E_{ \mathcal I}\right) 
	\leq \sum_{ \mathcal I \subseteq {\mathbf{S}_T: u_L \in \mathcal I}} P(E_{ \mathcal I}),
	\end{aligned} 
	\end{equation}
	and we require that $P_e (\mathbf S, T) \leq \epsilon 2^{-\nu n}$. Thus, for every $\mathcal I \subseteq \mathbf S_T$ such that $u_L \in \mathcal I$, it suffices to show that $P(E_\mathcal I) \leq \epsilon 2^{-(L-1)} 2^{-\nu n}$. \\
\begin{figure}[h]
	\centering
	\includegraphics[width=.65\textwidth,height=.25\textheight]{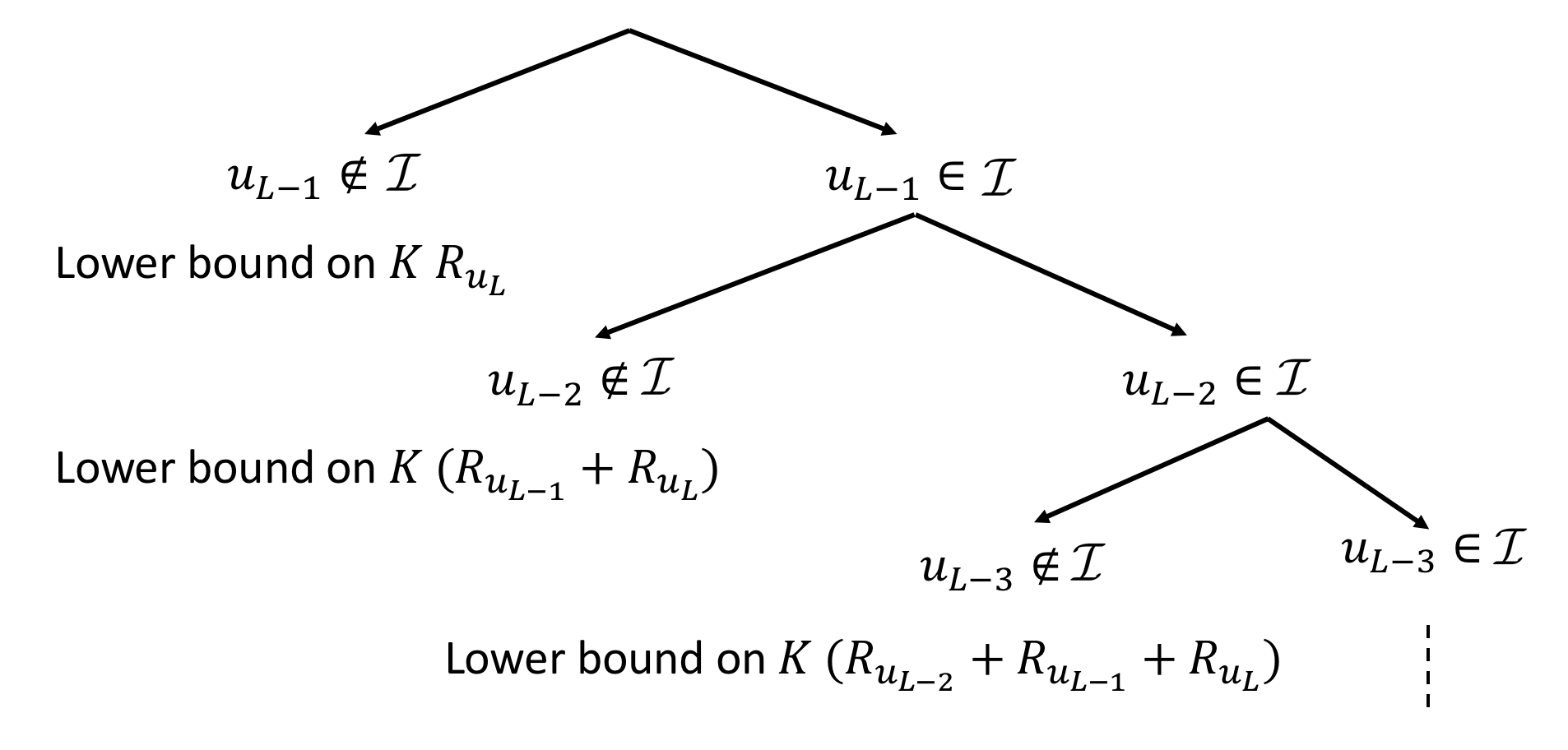}
	\caption{Illustration of the error analysis of Construction \ref{Random binning multi-version code}. \color{black} \label{Recursive Argument}}
    \end{figure}	
	\noindent We now proceed in a case by case manner as shown in Fig. \ref{Recursive Argument}. We first consider the case where $u_{L-1} \notin \mathcal I$, later we consider the case where $u_{L-1} \in \mathcal{I}$. For the case where $u_{L-1} \notin \mathcal I$, we have
	\begin{align}
	E_{ \mathcal I} \subset \tilde E_{u_{L-1}} &\coloneqq \{\exists \mathbf w_{u_L}' \neq \mathbf W_{u_L}: \varphi_{u_L}(\mathbf w_{u_L}')=\varphi_{u_L}(\mathbf W_{u_L}) \ \text{and} \ (\mathbf W_{u_{L-1}}, \mathbf w_{u_L}') \in A_{\delta_K} \}. 
	\end{align}
	\noindent Consequently, we have
	\begin{align}
	P(E_{ \mathcal I}) & \leq P(\tilde E_{u_{L-1}}) 
= \sum_{(\mathbf w_{u_{L-1}}, \mathbf w_{u_L})} p(\mathbf w_{u_{L-1}}, \mathbf w_{u_L}) \notag \\ &P( \exists \mathbf w_{u_L}' \neq \mathbf w_{u_L}: \varphi_{u_L}(\mathbf w_{u_L}')=\varphi_{u_L}(\mathbf w_{u_L}) \ \text{and} \ (\mathbf w_{u_{L-1}}, \mathbf w_{u_L}') \in A_{\delta_K}) \notag \\
	& \stackrel{(a)}{\leq} \sum_{(\mathbf w_{u_{L-1}}, \mathbf w_{u_L})} p(\mathbf w_{u_{L-1}}, \mathbf w_{u_L})\sum_{\substack{\mathbf w_{u_L}' \neq \mathbf w_{u_L}\\ (\mathbf w_{u_{L-1}}, \mathbf w_{u_L}') \in A_{\delta_K}}} P(\varphi_{u_L}(\mathbf w_{u_L}')=\varphi_{u_L}(\mathbf w_{u_L})) \notag \\
	&  \stackrel{(b)}{=} \sum_{(\mathbf w_{u_{L-1}}, \mathbf w_{u_L})} p(\mathbf w_{u_{L-1}}, \mathbf w_{u_L})\sum_{\substack{\mathbf w_{u_L}' \neq \mathbf w_{u_L}\\ (\mathbf w_{u_{L-1}}, \mathbf w_{u_L}') \in A_{\delta_K}}} \prod_{i=1}^{c} P(\varphi_{u_L}^{(t_i)}(\mathbf w_{u_L}')=\varphi_{u_L}^{(t_i)}(\mathbf w_{u_L})) \notag \\
	&\leq \sum_{(\mathbf w_{u_{L-1}}, \mathbf w_{u_L})} p(\mathbf w_{u_{L-1}}, \mathbf w_{u_L})  |A_{\delta_K} (\mathbf W_{u_L}|\mathbf w_{u_{L-1}})| \ \prod_{i=1}^{c} 2^{-K R_{u_L}^{(t_i)}/c }
	 \notag \\
	&\stackrel{(c)}{=} \sum_{(\mathbf w_{u_{L-1}}, \mathbf w_{u_L})} p(\mathbf w_{u_{L-1}}, \mathbf w_{u_L})  |A_{\delta_K} (\mathbf W_{u_L}|\mathbf w_{u_{L-1}})| \ 2^{-K R_{u_L} }  \notag \\ 
	&= \sum_{(\mathbf w_{u_{L-1}}, \mathbf w_{u_L})} p(\mathbf w_{u_{L-1}}, \mathbf w_{u_L})  2^{-(KR_{u_L} -\log Vol ((u_L-u_{L-1}) \delta_K K, K))} \notag \\
	& = 2^{-(KR_{u_L} -\log Vol ((u_L-u_{L-1}) \delta_K K, K))},
	\end{align} 
	where $(a)$ follows by the union bound, $(b)$ follows since each server assigns an index independently from the other servers and $(c)$ follows from (\ref{sum_rate_all_servers}). Choosing $R_{u_L}$ to satisfy 
	\begin{align*}
	KR_{u_L} \geq \log Vol ((u_L-u_{L-1}) \delta_K K, K)+(L-1)- \log \epsilon 2^{-\nu n}
	\end{align*}
	ensures that $P(E_\mathcal I) \leq \epsilon 2^{-(L-1)} 2^{-\nu n}$.

	Now, we consider the case where $u_{L-1} \in \mathcal I$. In this case, we consider the following two cases.
	First, we consider the case where $u_{L-2} \notin \mathcal I $, later we consider the case where $u_{L-2} \in \mathcal{I}$. For the case where $u_{L-2} \notin \mathcal I,$ we have 
	\begin{align}
	E_\mathcal I \subseteq \tilde E_{u_{L-2}} &\coloneqq \{\exists \mathbf w_{u_{L-1}}' \neq \mathbf W_{u_{L-1}}, \mathbf  w_{u_L}' \neq \mathbf W_{u_L}:  \varphi_{u_{L-1}}(\mathbf  w_{u_{L-1}}')=\varphi_{u_{L-1}}(\mathbf  W_{u_{L-1}}), \notag \\ &\varphi_{u_L}(\mathbf  w_{u_L}')=\varphi_{u_L}(\mathbf  W_{u_L}) \ \text{and} \ (\mathbf  W_{u_{L-2}}, \mathbf w_{u_{L-1}}', \mathbf w_{u_L}') \in A_{\delta_K} \}. 
	\end{align}
	Therefore, we have 
	\begin{align}
	 P(E_\mathcal I) &< P(\tilde E_{u_{L-2}}) \leq  \sum_{(\mathbf  w_{u_{L-2}}, \mathbf w_{u_{L-1}}, \mathbf w_{u_L})} p(\mathbf  w_{u_{L-2}}, \mathbf w_{u_{L-1}}, \mathbf w_{u_L}) \notag \\ &\sum_{\substack{\mathbf w_{u_{L-1}}' \neq \mathbf w_{u_{L-1}}\\ \mathbf w_{u_{L}}' \neq \mathbf w_{u_{L}} \\ (\mathbf w_{u_{L-2}}, \mathbf w_{u_{L-1}'}, \mathbf w_{u_{L}}') \in A_{\delta_K}}} P(\varphi(\mathbf  w_{u_{L-1}}')=\varphi(\mathbf  w_{u_{L-1}})) P(\varphi_{u_L}(\mathbf  w_{u_L}')=\varphi(\mathbf  w_{u_L})) \notag \\
	&\leq \sum_{(\mathbf  w_{u_{L-2}}, \mathbf w_{u_{L-1}}, \mathbf w_{u_L})} 
	p(\mathbf  w_{u_{L-2}}, \mathbf w_{u_{L-1}}, \mathbf w_{u_L})  2^{-K(R_{u_{L-1}}+R_{u_L})}  |A_{\delta_K} (\mathbf W_{u_{L-1}}, \mathbf W_{u_{L}}|\mathbf w_{u_{L-2}})| \notag \\
	&= \sum_{(\mathbf  w_{u_{L-2}}, \mathbf w_{u_{L-1}}, \mathbf w_{u_L})} 
	p(\mathbf  w_{u_{L-2}}, \mathbf w_{u_{L-1}}, \mathbf w_{u_L})  2^{-K(R_{u_{L-1}}+R_{u_L})} \notag \\ & \ \ \ \ \ 2^{\log Vol ((u_L-u_{L-1}) \delta_K K, K)+\log Vol ( (u_{L-1}-u_{L-2}) \delta_K K, K)}.
	\end{align} 
	In this case, we choose the rates as follows
	\begin{align*}
	K(R_{u_{L-1}}+R_{u_L})  \geq \sum_{j=L-1}^{L} \log Vol ( (u_{j}-u_{j-1}) \delta_K K, K)+(L-1)-\log \epsilon 2^{-\nu n}.
	\end{align*}	
	We next consider the other case where $u_{L-2} \in \mathcal I$. In this case, we also have two cases based on whether $u_{L-3}$ is in $\mathcal I$ or not. By applying the above argument repeatedly, we obtain the following conditions for the overall probability of error to be upper bounded by $\epsilon 2^{-\nu n}$. 
	\begin{align}
	\label{Rate Region}
	K \sum_{j=i}^{L} R_{u_j} & \geq  \sum_{j=i}^{L} \log Vol ( (u_j-u_{j-1}) \delta_K K, K)+(L-1)-\log \epsilon 2^{-\nu n},  \forall i \in \{2, 3, \cdots, L\}, \notag \\
	K \sum_{j=1}^{L}  R_{u_j} & \geq K+\sum_{j=2}^{L} \log Vol ((u_{j}-u_{j-1}) \delta_K K, K)+(L-1)-\log \epsilon 2^{-\nu n}.
	\end{align}
	Since $\log Vol (m \delta_K K, K) \leq m \log Vol (\delta_K K, K), \forall m \in \mathbb{Z}^{+}$, it suffices if the rates satisfy  
	\begin{align}
	\label{relaxed rate region}
	K \sum_{j=i}^{L} R_{u_j} & \geq  \sum_{j=i}^{L}  (u_j-u_{j-1}) \log  Vol (\delta_K K, K)+(L-1)-\log \epsilon 2^{-\nu n},  \forall i \in \{2, 3, \cdots, L\}, \notag \\
	K \sum_{j=1}^{L}  R_{u_j} & \geq K+\sum_{j=2}^{L} (u_{j}-u_{j-1}) \log  Vol (\delta_K K, K)+(L-1)-\log \epsilon 2^{-\nu n}.
	\end{align}
	
	The rates chosen in (\ref{eq:SWrate1}), (\ref{eq:SWrate2}) satisfy the above inequalities, therefore our construction has a probability of error bounded by $\epsilon 2^{-\nu n}$. Moreover, as in \cite[Chpater 7]{el2011network}, it follows that there exists a deterministic code that has a probability of error bounded by $\epsilon$. The worst-case storage cost is when a server receives all versions and is given by
	\begin{align*}
	\frac{K-\log \epsilon 2^{-\nu n}}{c}+\frac{(\nu -1) (\log Vol(\delta_K K, K)-\log \epsilon 2^{-\nu n}+\nu/2)}{c}.
\end{align*}
\end{proof}
\noindent Motivated by the fact that linear codes have lower complexity, in the Appendix, we show that linear codes exist that achieve the storage cost of Theorem \ref{Slepian-Wolf Storage Cost Theorem}. Our proof is inspired by \cite{wyner1974recent}.
\begin{remark}
	The proof of Theorem \ref{Slepian-Wolf Storage Cost Theorem} uses simultaneous non-unique decoding ideas \cite{Bandemer_nonunique} used in several multi-user scenarios. In particular, with our non-unique decoding approach to decode $ \mathbf W_{u_L},$ the decoder picks the unique $\mathbf{w}_{u_L}$ such that $(\mathbf w_{u_1}, \mathbf w_{u_2},\ldots, \mathbf w_{u_L}) \in A_{\delta_K}$ for \emph{some} $\mathbf w_{u_1}, \mathbf w_{u_2},\ldots, \mathbf w_{u_{L-1}}$, which are consistent with the bin indices. We use this strategy since unlike the Slepian-Wolf problem where all the messages are to be decoded, we are only required to decode the latest common version. In contrast, the \emph{unique decoding} approach employed by Slepian-Wolf coding would require the decoder to obtain for some subset $S\subseteq \{u_1, u_2,\ldots, u_L\}$ such that $u_L \in S$, the unique $\mathbf w_{S}$ in the possible set that is consistent with the bin-indices; unique decoding, for instance, would not allow for correct decoding if there are multiple possible tuples even if they happen to have the same latest common version $\mathbf{w}_{u_L}$. The discussion in \cite{Bidhokti_nonunique}, which examined the necessity of non-unique decoding, motivates the following question: Can we use the decoding ideas of Slepian-Wolf - where all the messages are decoded - however, for an appropriately chosen subset of versions and have the same rates? In other words, if we take the union of the unique decoding rate regions over all possible subsets of $\{\mathbf W_{u_1}, \mathbf W_{u_2}, \ldots, \mathbf W_{u_L}\}$, does the rate allocation of (\ref{eq:SWrate1}), (\ref{eq:SWrate2}) lie in this region? The answer of this question is that non-unique decoding provides better rates than unique decoding in our case as we explain in Example \ref{counterexample}. 
   \begin{example}
   	\label{counterexample}
	We consider the case where $\delta_K = \delta,$  $c=2$ and $\nu=3$. Consider the state where server $2$ does not receive $\mathbf W_1$. Then, the storage allocation of our scheme is given by Table \ref{table: Example}. We use $KR_u$ to denote the total number of bits stored for version $u \in [\nu]$ in Table \ref{table: Example}.
	\begin{table*}[h]
		\renewcommand{\arraystretch}{1.2}
		\centering
		\begin{tabular}{ |p{2 cm} | p{3 cm} | p{3 cm} | }
			\hline
			\textbf{Version} &\textbf{Server 1} & \textbf{Server 2} \\ \hline
			$\mathbf W_1$ & $\frac{K+o(K)}{2}$ &-
			\\ \hline	
			$\mathbf W_2$ & $\frac{KH(\delta)+o(K)}{2}$ & $\frac{K+KH(\delta)+o(K)}{2}$
			\\ \hline	
			$\mathbf W_3$ & $\frac{KH(\delta)+o(K)}{2}$  & $\frac{KH(\delta)+o(K)}{2}$
			\\ \hline
		\end{tabular}
		\caption{Storage Allocation of Example $1$.}\label{table: Example}
	\end{table*}

We first examine unique decoding based decoders that aim to recover $\mathbf W_{3}$. It is clear that the decoder cannot recover the $K$ bits of $\mathbf W_{3}$ without using side information, since the total number of bits of $\mathbf W_3$ stored is only $KH(\delta)+o(K)$.

We next consider the case where a unique decoding based decoder uses the subset $\{\mathbf W_1, \mathbf W_3\}$.  The rates $(R_1^{(unique, \mathbf W_1, \mathbf W_3)}, R_3^{(unique, \mathbf W_1, \mathbf W_3)})$ for a vanishing probability of error must satisfy 
\begin{align*}
K R_3^{(unique, \mathbf W_1, \mathbf W_3)} &\geq K H(\delta* \delta)+o(K),  \\
K( R_1^{(unique, \mathbf W_1, \mathbf W_3)}+R_3^{(unique, \mathbf W_1, \mathbf W_3)}) &\geq K+ K H(\delta* \delta) + o(K),
\end{align*}
where $\delta * \delta = 2 \delta(1-\delta)$. However, $K R_3 = K H(\delta)+o(K) < KR_{3}^{(unique, \mathbf W_1, \mathbf W_3)}$.

We next consider the case where a unique decoding based decoder uses $\{\mathbf W_2,\mathbf W_3\}$ for decoding.  In this case, the decoder requires
\begin{align*}
K R_3^{(unique, \mathbf W_2, \mathbf W_3)}  &\geq  KH(\delta) +o(K), \\
K(R_2^{(unique, \mathbf W_2, \mathbf W_3)}+ R_3^{(unique, \mathbf W_2, \mathbf W_3)}) &\geq K+ KH(\delta) + o(K)
\end{align*}
We notice $K(R_2+R_3) = K/2 + 2KH(\delta)+o(K)< K(R_2^{(unique, \mathbf W_2, \mathbf W_3)}+ R_3^{(unique, \mathbf W_2, \mathbf W_3)})$, for $\delta < H^{-1}(0.5).$ 

Finally, consider the case where a unique decoding based decoder uses all three messages $\{\mathbf{W}_{1},\mathbf W_2,\mathbf W_3\}.$ In this case, the rate tuples $(R_1^{(unique, \mathbf W_1, \mathbf W_2, \mathbf W_3)}, R_3^{(unique, \mathbf W_1, \mathbf W_2, \mathbf W_3)})$ have to satisfy seven inequalities, including the following inequality
$$K( R_1^{(unique, \mathbf W_1, \mathbf W_2, \mathbf W_3)}+R_3^{(unique, \mathbf W_1, \mathbf W_2, \mathbf W_3)}) \geq K+ K H(\delta* \delta) + o(K)$$
Clearly $K(R_1 + R_3) <  K(R_1^{(unique, \mathbf W_1, \mathbf W_2, \mathbf W_3)}+R_3^{(unique, \mathbf W_1, \mathbf W_2, \mathbf W_3)}).$ 
Thus, the union of the unique decoding rate regions for vanishing error probabilities, taken over all possible subsets of $\{\mathbf W_1, \mathbf W_2, \mathbf W_3\},$ does not include the rate tuple of Table \ref{table: Example}. 
\end{example}
\end{remark}

\begin{remark}
\label{rem:zeroerror}
{
The result of Theorem \ref{Slepian-Wolf Storage Cost Theorem} allows for erroneous decoding with probability at most $\epsilon.$ A natural question is to ask whether a similar storage cost can be obtained if we want the probability of error to be $0$, i.e., if we want our decoder to correctly decode for \emph{every} possible message tuple. The answer to this question in general is connected to the question of whether $0$-error rate region and $\epsilon$-error rate region are identical for our setting. There are several instances in network information theory where, even though there is no noise in the network the $\epsilon$-error capacity is still larger than the zero error capacity, (e.g., multiple access \cite{cover_thomas}). For some networks, the answer to this question is unknown and involves deep connections to other related questions \cite{langberg2011network}. Note also for distributed source coding setups such as our problem, the determination of the $0$-error capacity is more complicated and involves the use of graph entropy \cite{korner1973coding}. In this paper, we leave the question of whether the $0$-error and $\epsilon$-error rate regions are the same, open.}
\end{remark}

\section{Lower Bound on The Storage Cost (Theorem \ref{Lower Bound Theorem})}
\label{Lower Bound on the storage cost}
In this section, we extend the lower bound on the per-server storage cost of \cite{multi_arxiv} for the case where we have correlated versions, and we require the probability of error to be at most $\epsilon$.
\begin{theorem}
	\label{Lower Bound Theorem}[Storage Cost Lower Bound]
	An $\epsilon$-error $(n, c, \nu, 2^K, q, \delta_K)$ multi-version code with correlated versions 
	such that $\mathbf W_{1}\rightarrow \mathbf W_{2} \rightarrow \ldots \rightarrow \mathbf W_{\nu}$ form a Markov chain, $\mathbf W_{m} \in [2^K]$  and given $\mathbf W_{m}$, $\mathbf W_{m+1}$ is uniformly distributed in a Hamming ball of radius $\delta_K K$ centered around $\mathbf W_{m}$ must satisfy 
	\begin{align}
	\log q \geq \frac{K+(\nu-1) \log Vol (\delta_K K, K)}{c+\nu-1} +
	\frac{\log (1-\epsilon 2^{\nu n})-\log \binom{c+\nu-1}{\nu} \nu!}{c+\nu-1},
	\end{align}
	where $\epsilon<1/2^{\nu n}$.	
\end{theorem}
\begin{proof}[Proof of Theorem \ref{Lower Bound Theorem} for $\nu=2$]
	Consider any $\epsilon$-error $(n,c,2, 2^K, q, \delta_K)$ multi-version code, and consider the first $c$ servers, $T=[c]$, for decoding. We recall that the set of possible tuples $A_{\delta_K}$ is partitioned into disjoint sets as follows
   \begin{align*}
A_{\delta_K}= A_{\delta_K, 1} \cup A_{\delta_K, 2},
\end{align*}
	where $A_{\delta_K, 1}$ is the set of tuples $(\mathbf{w}_1, \mathbf{w}_2) \in A_{\delta_K}$ for which we can decode successfully for all $\mathbf S  \in \mathcal{P}([\nu])^n$ and $A_{\delta_K, 2}$ is the set of tuples where we cannot decode successfully at least for one state $\mathbf S  \in \mathcal{P}([\nu])^n$, which can be expressed as follows   
	\begin{align}
	A_{\delta_K, 2}= \bigcup_{\mathbf{S} \in \mathcal{P}([\nu])^n} A_{\delta_K, 2}^{(\mathbf S)},
	\end{align}
	where  $A_{\delta_K, 2}^{(\mathbf S)}$ is the set of tuples for which we cannot decode successfully given a particular state $\mathbf S \in \mathcal{P}([\nu])^n$. Consequently, we have
	\begin{align}
	|A_{\delta_K, 2}| \leq \sum\limits_{\mathbf S \in \mathcal{P}([\nu])^n} |A_{\delta_K, 2}^{(\mathbf S)}|.
	\end{align}
	For any state $\mathbf S \in \mathcal{P}([\nu])^n$, we require the probability of error, $P_e$, to be at most $\epsilon$. Since all tuples in the set $A_{\delta_K}$ are equiprobable, we have
	\begin{align}
	P_e=\frac{|A_{\delta_K, 2}^{(\mathbf S)}|}{|A_{\delta_K}|} \leq \epsilon.
	\end{align}
	Therefore, we have
	\begin{align}
	\label{lower bound on A1}
	|A_{\delta_K, 1}|&=|A_{\delta_K}|-|A_{\delta_K, 2}| \notag \\ 
	& \geq  |A_{\delta_K}|-\sum_{\mathbf S \in \mathcal{P}([\nu])^n} |A_{\delta_K, 2}^{(\mathbf S)}| \notag  \\ & \geq |A_{\delta_K}|-\sum_{\mathbf S \in \mathcal{P}([\nu])^n}  \epsilon  |A_{\delta_K}| \notag \\
	& \geq   |A_{\delta_K}|(1-\epsilon  2^{\nu n} ).
	\end{align}   
	Suppose that $(\mathbf W_1, \mathbf W_2) \in A_{\delta_K, 1}.$ Because of the decoding requirements, if $\mathbf W_1$ is available at all servers, the decoder must be able to obtain $\mathbf W_1$ and if $\mathbf W_2$ is available at all servers, the decoder must return $\mathbf W_2$. Hence, as shown in  \cite{multi_arxiv}, there exist $\mathbf{S}_1, \mathbf{S}_2 \in \mathcal{P}([\nu])^n$ such that 
	\begin{itemize}
	\item $\mathbf{S}_1$ and $\mathbf{S}_2$ differ only in the state of one server indexed by $B \in [c]$, and
	\item $\mathbf W_1$ can be recovered from the first $c$ servers in state $\mathbf{S}_1$ and $\mathbf W_2$ can be recovered from the first $c$ servers in $\mathbf{S}_2$.
	\end{itemize}
	Therefore both $\mathbf W_1$ and $\mathbf W_2$ are decodable from the $c$ codeword symbols of the first $c$ servers in state $\mathbf{S}_1$, and the codeword symbol of the $B$-th server in state $\mathbf{S}_2$. Thus, we require the following 
	\begin{align}
	c \ q^{c+1} &\geq |A_{\delta_K, 1}|  \notag \\ & \geq |A_{\delta_K}|(1-\epsilon  2^{\nu n} ).
	\end{align}
We also have $|A_{\delta_K}| = 2^K Vol(\delta_K K, K)$. Therefore, the storage cost is lower-bounded as follows
	\begin{align}
	\log q  \geq \frac{K+ \log  Vol(\delta_K K, K)}{c+1}+\frac{\log (1-\epsilon  2^{\nu n} )-\log c}{c+1}.    
	\end{align}
\end{proof}

\label{proof for any number of versions}
\noindent We now provide a proof sketch for the case where $\nu \geq 3$.
\begin{proof}[Proof sketch of Theorem \ref{Lower Bound Theorem} for $\nu \geq 3$] Consider any $\epsilon$-error $(n, c, \nu, 2^K, q, \delta_K)$ multi-version code, and consider the first $c \leq n$ servers, $T=[c]$, for decoding.  Suppose we have $\nu$ versions $\mathbf W_{[\nu]}=(\mathbf W_1, \mathbf W_2, \cdots, \mathbf W_{\nu})$.  Suppose that $\mathbf W_{[\nu]} \in A_{\delta_K, 1}$. We construct auxiliary variables $Y_{[c-1]}, Z_{[\nu]}, B_{[\nu]}$, where $Y_i, Z_j \in [q], i \in [c-1], j \in [\nu]$  , $1 \leq B_1 \leq \cdots \leq B_{\nu} \leq c$ and a permutation $\Pi: [\nu] \rightarrow [\nu]$, such that there is a bijection mapping from these variables to $A_{\delta_K, 1}$. In order to construct these auxiliary variables, we use the algorithm of \cite{multi_arxiv}. 
	Therefore, we have 
	\begin{align}
	q^{c+\nu-1} \binom{c+\nu-1}{\nu} \nu! 
	& \geq |A_{\delta_K, 1}| \notag \\ &> |A_{\delta_K}|(1-\epsilon  2^{\nu n} ),
	\end{align}
	where the first inequality follows since $Y_i, Z_j \in [q]$, there are at most $\binom{c+\nu-1}{\nu}$ possibilities of $B_{[\nu]}$ and at most $\nu!$ possible permutations. We also have $|A_{\delta_K}| = 2^K Vol(\delta_K K, K)^{(\nu-1)}$. Therefore, the storage cost is lower-bounded as follows  
	\begin{align}
	\log q \geq \frac{K+(\nu-1) \log Vol (\delta_K K, K)}{c+\nu-1} +
	\frac{\log (1-\epsilon 2^{\nu n})-\log \binom{c+\nu-1}{\nu} \nu!}{c+\nu-1}. 
	\end{align}
\end{proof}



\color{black}

\section*{Acknowledgement}
The authors would like to thank A. Tulino and J. Llorca for their helpful comments.
\section{Conclusion}  
\label{Conclusion}
In this paper, we have proposed multi-version codes to efficiently store correlated updates of data in a decentralized asynchronous storage system. These constructions are based on Reed-Solomon codes and random binning. An outcome of our results is that the correlation between versions can be used to reduce storage costs in asynchronous decentralized systems, even if there is no single server or client node who is aware of all data versions, in applications where consistency is important. In addition, our converse result shows that these constructions are within a factor of $2$ from the information-theoretic optimum in certain interesting regimes. The development of practical coding schemes for the case where $\delta_K$ is known a priori is an open research question, which would require non-trivial generalizations of previous code constructions for the Slepian-Wolf problem \cite{pradhan2003distributed, schonberg2004distributed}.

\appendix
\label{sec:appendix}

In this appendix, we show that there exist \emph{linear} codes that achieve the storage cost of Theorem \ref{Slepian-Wolf Storage Cost Theorem}. The proof uses linear binning instead of random binning, but mirrors the random binning proof in other respects and we focus on the key differences here.  
\begin{lemma}
\label{lemma:linearbinning}
Let $G$ be an $N \times M$ matrix whose entries are chosen according to Bernoulli($p$) independently of each other. Let $\mathbf{u}$ be any non-zero $N \times 1$ vector. We have 
\begin{align}
\mathbb{P}(\mathbf{u} ^{\rm T} G= 0) = ((1+(1-2p)^{w_H(u)})/2)^M.
\end{align}
\end{lemma}
\begin{proof}
 Consider $k$ Bernoulli trials where the probability of success of each trial is $p$. It can be shown that an even number of successes among the $k$ trials occurs with probability $(1+(1-2p)^k)/2.$
Therefore, we have $\mathbb{P}(\mathbf{u}^{\rm T} G= 0)=((1+(1-2p)^{w_H(u)})/2)^M.$
\end{proof}
\noindent We next explain the code construction and provide an alternate proof for Theorem \ref{Slepian-Wolf Storage Cost Theorem}.
\begin{construction}[Random Linear binning multi-version code] 
Suppose that the $i$-th server receives the versions $\mathbf S(i)=\{s_1, s_2, \cdots, s_{|\mathbf S(i)|}\} \subseteq [\nu]$, where $s_1<s_2<\cdots<s_{|\mathbf S(i)|}$.
		\begin{itemize}
			\item  \emph{Random code generation:} At the $i$-th server, for version $s_j$ the encoder creates a random binary matrix $G_{s_j}^{(i)}$ with $K$ rows and $(K+(\nu-1) \log Vol(\delta_K K, K)+(\nu-1)-\log \epsilon 2^{-\nu n})/c$ columns, where each entry is chosen as Bernoulli($1/2$) independently of all the other entries in the matrix and all other matrices. We denote by ${G}_{s_j, m}^{(i)}$ the first $m$ columns of ${G}_{s_j}^{(i)}$.

			\item \emph{Encoding}: The server stores $\mathbf{W}_{s_j}^{\rm T} {G}_{s_j, KR_{s_j}^{(i)}/c}^{(\ell)}$ for version $s_j$, where  $R_{s_j}^{(i)}/c$ is the rate assigned by the $i$-th server to version $s_j$. The decoder is also aware of the matrix ${G}_{s_j}^{(i)}$ a priori. The encoding function of the $i$-th server is defined as follows
			\begin{align}
			\varphi_S^{(i)}=(\mathbf{W}_{s_1}^{\rm T}{G}^{(i)}_{s_1, KR_{s_1}^{(i)}/c}, \mathbf{W}_{s_2}^{\rm T} {G}^{(i)}_{s_2, KR_{s_2}^{(i)}/c}, \ldots, \mathbf{W}_{s_{|\mathbf S(i)|}}^{\rm T} {G}^{(i)}_{s_{|\mathbf S(i)|}, KR_{s_{|\mathbf S(i)|}}^{(i)}/c} ),
			\end{align}
where we choose the rates as given by (\ref{eq:SWrate1}), (\ref{eq:SWrate2}). 

			\item \emph{Decoding}:
			Consider a state $\mathbf{S} \in \mathcal{P}([\nu])^n$ and assume that the decoder connects to the servers $T=\lbrace t_1, t_2, \cdots, t_c \rbrace$. Let $\mathbf W_{u_L}$ be the latest common version among these servers and that the versions $\mathbf W_{u_1}, \mathbf W_{u_2}, \cdots, \mathbf W_{u_{L-1}}$ are the older versions such that each is received by at least one server out of those $c$ servers. This set of versions is denoted by $\mathbf S_{\rm T}$ and defined in (\ref{set of versions}). Given the bin indices $(b_{u_1}, b_{u_2}, \cdots, b_{u_L})$, the decoder finds all tuples $(\mathbf w_{u_1}, \mathbf w_{u_2}, \cdots, \mathbf w_{u_L}) \in A_{\delta_K}$ such that $(\varphi_{u_1} (\mathbf w_{u_1})=b_{u_1}, \varphi_{u_2} (\mathbf w_{u_2})=b_{u_2}, \cdots, \varphi_{u_L} (\mathbf w_{u_L})=b_{u_L})$. If all tuples have the same latest common version $\mathbf w_{u_L}$, the decoder declares $\mathbf w_{u_L}$ to be the estimate of the latest common version $\mathbf{\hat W_{u_L}}$. Otherwise, the decoder declares an error. 
		\end{itemize}
\end{construction}
\begin{proof}[Proof of Theorem \ref{Slepian-Wolf Storage Cost Theorem} using linear binning]
The probability of error in decoding the latest common version among the $c$ servers is upper-bounded as follows
\begin{equation*}
P_e (\mathbf S, T)= P(E)  =P \left( \bigcup_{ \mathcal I \subseteq \mathbf {S}_T: u_L \in \mathcal I} E_{ \mathcal I}\right) 
\leq \sum_{ \mathcal I \subseteq {\mathbf{S}_T: u_L \in \mathcal I}} P(E_{ \mathcal I}),
\end{equation*}
	and we require that $P_e (\mathbf S, T) \leq \epsilon 2^{-\nu n}$. We proceed in a case by case manner  and first consider the case where $u_{L-1} \notin \mathcal I$, later we consider the case where $u_{L-1} \in \mathcal{I}$. For the case where $u_{L-1} \notin \mathcal I$, we have the following
	\begin{align}
	E_{ \mathcal I} \subset \tilde E_{u_{L-1}} &\coloneqq \{\exists \mathbf w_{u_L}' \neq \mathbf W_{u_L}:  \mathbf w_{u_L}'^{\rm T} {G}^{(t_i)}_{u_L, KR_{u_L}^{(t_i)}/c} = \mathbf W_{u_L}^{\rm T}{G}^{(t_i)}_{u_L, KR_{u_L}^{(t_i)}/c}, \forall i\in[c] \notag \\ & ~~~~~~\text{and} \ (\mathbf W_{u_{L-1}}, \mathbf w_{u_L}') \in A_{\delta_K} \}. 
	\end{align}
	\noindent Consequently, we have
	$P(E_{ \mathcal I}) < P(\tilde E_{u_{L-1}}),$
	and we can upper-bound $P(\tilde E_{u_{L-1}})$ as follows
	\begin{align}
	P(\tilde E_{u_{L-1}}) &= \sum_{(\mathbf w_{u_{L-1}}, \mathbf w_{u_L})} p(\mathbf w_{u_{L-1}}, \mathbf w_{u_L}) \\&P\left( \exists \mathbf w_{u_L}' \neq \mathbf w_{u_L}: \mathbf w_{u_L}'^{\rm T}{G}^{(t_i)}_{u_L,  KR_{u_L}^{(t_i)}/c} =\mathbf w_{u_L}^{\rm T}{G}^{(t_i)}_{u_L, KR_{u_L}^{(t_i)}/c}, \forall i\in[c], \ (\mathbf w_{u_{L-1}}, \mathbf w_{u_L}') \in A_{\delta_K}\right) \notag \\
	& \stackrel{(a)}{\leq} \sum_{(\mathbf w_{u_{L-1}}, \mathbf w_{u_L})} p(\mathbf w_{u_{L-1}}, \mathbf w_{u_L})\sum_{\substack{\mathbf w_{u_L}' \neq \mathbf w_{u_L}\\ (\mathbf w_{u_{L-1}}, \mathbf w_{u_L}') \in A_{\delta_K}}} \prod_{i=1}^{c}P\left( (\mathbf w_{u_L}'+ \mathbf w_{u_L})^{\rm T}{G}^{(t_i)}_{u_L, KR_{u_L}^{(t_i)}/c}=0\right) \notag  \\
	& \stackrel{(b)}{=} \sum_{(\mathbf w_{u_{L-1}}, \mathbf w_{u_L})} p(\mathbf w_{u_{L-1}}, \mathbf w_{u_L}) Vol ((u_L-u_{L-1}) \delta_K K, K) \prod_{i=1}^{c}2^{-KR^{(t_i)}_{u_L}/c} \notag \\
	& = 2^{-(KR_{u_L} -\log Vol ((u_L-u_{L-1}) \delta_K K, K))}.
	\end{align} 
where $(a)$ follows since the matrices ${G}^{(t_1)}, {G}^{(t_2)},\ldots, {G}^{(t_c)}$ are chosen independently and $(b)$ follows from Lemma \ref{lemma:linearbinning}. Choosing $R_{u_L}$ to satisfy 
	$KR_{u_L} \geq \log Vol ((u_L-u_{L-1}) \delta_K K, K)+(L-1)- \log \epsilon 2^{-\nu n}$ 
	ensures that $P(E_\mathcal I) \leq \epsilon 2^{-(L-1)} 2^{-\nu n}$. \\ Now, we consider the case where $u_{L-1} \in \mathcal I$. In this case, we consider the following two cases.
	First, we consider the case where $u_{L-2} \notin \mathcal I $, later we consider the case where $u_{L-2} \in \mathcal{I}$. For the case where $u_{L-2} \notin \mathcal I,$ we have 
	\begin{align}
	E_\mathcal I \subseteq \tilde E_{u_{L-2}} &\coloneqq \{\exists \mathbf w_{u_{L-1}}' \neq \mathbf W_{u_{L-1}}, \mathbf  w_{u_L}' \neq \mathbf W_{u_L}:  \mathbf  w_{u_{L-1}}'^{\rm T} {G}^{(t_i)}_{u_{L-1}, KR_{u_{L-1}}^{(t_i)}/c}= \mathbf  W_{u_{L-1}}^{\rm T}{G}^{(t_i)}_{u_{L-1}, KR_{u_{L-1}}^{(t_i)}/c}, \notag \\ &~~~~\mathbf  w_{u_L}'^{\rm T}{G}^{(t_i)}_{u_{L}, KR_{u_{L}}^{(t_i)}/c}= \mathbf  W_{u_L}^{\rm T}{G}^{(t_i)}_{u_{L}, KR_{u_{L}}^{(t_i)}/c}  \ \text{and} \ (\mathbf  W_{u_{L-2}}, \mathbf w_{u_{L-1}}', \mathbf w_{u_L}') \in A_{\delta_K} \}. 
	\end{align}
	In this case, we choose the rates to satisfy $K(R_{u_{L-1}}+R_{u_L})  \geq \sum\limits_{j=L-1}^{L} \log Vol ( (u_{j}-u_{j-1}) \delta_K K, K)+(L-1)-\log \epsilon 2^{-\nu n}.$	
	By applying the above argument repeatedly, we obtain the region in (\ref{Rate Region}).
\end{proof}

}

\bibliographystyle
{IEEEtran}
\bibliography{IEEEabrv,Nulls}

\end{document}